%% file: main.tex
\documentclass[year=24,pdfa]{fmcad_arxiv}

\usepackage{graphicx}
\usepackage{hyperref}
\usepackage{amssymb}
\usepackage{booktabs}
\usepackage{multirow} 
\usepackage{colortbl}
\usepackage{verbatim}
\usepackage{array}
\usepackage{mathtools}
\usepackage{amssymb}
\usepackage{thmtools,thm-restate}
\usepackage{algorithm}
\usepackage[noend]{algpseudocode}
\usepackage{subcaption}
\usepackage[pdf]{graphviz}
\usepackage{stmaryrd}
\usepackage{proof}
\usepackage{mathpartir}
\usepackage{tlatex}
\usepackage{relsize}
\usepackage{nccmath}
\usepackage{multicol}
\usepackage{booktabs}
\usepackage{pifont}
\usepackage{amsthm}
\usepackage{lineno}
\usepackage{url}

\usepackage{breakurl}

\newcommand\ian[1]{\textcolor{blue}{ID:#1}}

\newcommand{\IndState}{\State\hspace{\algorithmicindent}}
\newcommand{\IndDubState}{\State\hspace{\algorithmicindent}\hspace{\algorithmicindent}}

\algnewcommand{\LeftComment}[1]{\Statex \(\triangleright\) #1}

\newcommand\abSemantics[1]{\llbracket #1 \rrbracket_{\alpha\beta}}
\newcommand\aSemantics[1]{\llbracket #1 \rrbracket_{\alpha}}
\newcommand\bSemantics[1]{\llbracket #1 \rrbracket_{\beta}}
\newcommand\abmodels{\vdash}

\newcommand\salpha{\hat{\alpha}} % symbolic alpha
\newcommand\aalpha{\hat{\alpha}} % in case i accidentally think ``abstract alpha'' instead of symbolic
\newcommand\calpha{\alpha} % concrete alpha
\newcommand\vbeta{\medmath{\beta}} % beta for state variables
\newcommand\TLA{TLA\textsuperscript{+}}
\newcommand\expr[1]{\epsilon\bigl( #1 \bigr)}
\newcommand\exprb[1]{\epsilon\left( #1 \right)}

\newcommand\tolts{\textsc{lts}}
\newcommand\errlts{\textsc{err}}
\newcommand\proplts{\textsc{prop}}
\newcommand\parsym{\mathbin{\!/\mkern-5mu/\!}}
\newcommand\parexp{\parallel}
\newcommand\fixpoint{\textsc{Fix}}
\newcommand\prop{\mathcal{P}}
\newcommand\totorder{\leqslant}

\newcommand\st{\ : \ }

\newcommand\tlastate[1]{
    \begin{array}{lll}
        #1
    \end{array}
}
\newcommand\sstack[1]{
    \begin{array}{c}
        #1
    \end{array}
}

\newtheorem{theorem}{Theorem}
\newtheorem{lemma}{Lemma}

\theoremstyle{definition}
\newtheorem{definition}{Definition}
\newtheorem{example}{Example}

\theoremstyle{remark}

\newenvironment{thmn}[1]
{\noindent\textbf{Theorem #1.} \begin{itshape}}
{\end{itshape}}
\newenvironment{lemn}[1]
{\noindent\textbf{Lemma #1.} \begin{itshape}}
{\end{itshape}}

\begin{document}

\title{Recomposition: A New Technique for Efficient Compositional Verification}

\makeatletter
\newcommand{\linebreakand}{%
  \end{@IEEEauthorhalign}
  \hfill\mbox{}\par
  \mbox{}\hfill\begin{@IEEEauthorhalign}
}
\makeatother

\author{
\IEEEauthorblockN{Ian Dardik}
\IEEEauthorblockA{\textit{Carnegie Mellon University} \\
Pittsburgh, PA, USA\\
idardik@andrew.cmu.edu}
\and
\IEEEauthorblockN{April Porter}
\IEEEauthorblockA{\textit{University of Maryland, College Park} \\
College Park, MD, USA \\
aporter3@terpmail.umd.edu}
\and
\IEEEauthorblockN{Eunsuk Kang}
\IEEEauthorblockA{\textit{Carnegie Mellon University} \\
Pittsburgh, PA, USA \\
eunsukk@andrew.cmu.edu}
}

\maketitle

\begin{abstract}
Compositional verification algorithms are well-studied in the context of model checking.
Properly selecting components for verification is important for efficiency, yet has received comparatively less attention.
In this paper, we address this gap with a novel compositional verification framework that focuses on component selection as an explicit, first-class concept.
The framework decomposes a system into components, which we then \textit{recompose} into new components for efficient verification.
At the heart of our technique is the \textit{recomposition map} that determines how recomposition is performed; the component selection problem thus reduces to finding a good recomposition map.
However, the space of possible recomposition maps can be large.
We therefore propose heuristics to find a small portfolio of recomposition maps, which we then run in parallel.
We have implemented our techniques in a model checker for the \TLA{} language.
In our experiments, we show that our tool achieves competitive performance with TLC--a well-known model checker for \TLA{}--on a benchmark suite of distributed protocols.
%\keywords{}
\end{abstract}

\graphicspath{{src/Figs/}}

\pagestyle{plain}
\input{src/Tex/10-introduction}

\input{src/Tex/12-running-ex}
\input{src/Tex/15-prelims}
\input{src/Tex/17-parallel-comp}

\input{src/Tex/20-overview}

\input{src/Tex/25-decomp-alg}

\input{src/Tex/30-static-reduction}

\input{src/Tex/40-comp-verif}
\input{src/Tex/45-correctness}

\input{src/Tex/47-parallel-alg}
\input{src/Tex/50-experiments}

\input{src/Tex/60-related-work}

\input{src/Tex/70-conclusion}

\section*{Acknowledgements}
The authors are most grateful to Changjian Zhang, Andy Hammer, and the anonymous reviewers for their helpful comments on earlier versions of this paper.
This project was supported by the U.S. NSF Award \#2144860.

\bibliographystyle{IEEEtran}
\bibliography{IEEEabrv,sources}

\appendices
\input{src/Tex/90-appendix}

\end{document}

%% file: src/Tex/10-introduction.tex
\section{Introduction}
Model checking is an important tool for software, protocol, and algorithm development.
\emph{Compositional verification} is a paradigm in which a system is decomposed into components, which are then verified using a divide-and-conquer algorithm.
To help model checking scale to large programs and specifications, compositional verification remains an important type of technique for combating the state explosion problem \cite{Giannakopoulou:2018handbook}.

Most research papers on compositional verification assume that the components are pre-determined and focus solely on verification algorithms \cite{alur:2005,Chen:2009,Cheung:1995,Cheung:1996,Cheung:1999,Cobleigh:2003,Graf:1990,Gupta:2007,Pasareanu:2008,Pasareanu:2006,Sabnani:1989,Tai:1993,Zheng:2010,Zheng:2014,Zheng:2015}.
However, \emph{component selection}--that is, determining the set of decomposed components and the order in which they are verified--can greatly impact performance, in terms of both run time and state space size.
Yet there are comparatively fewer model checking frameworks that investigate component selection, e.g. by automating decomposition \cite{Cobleigh:2006,Metzler:2008,Nam:2008}.
Unfortunately, research into automated decomposition has seen limited success thus far; as Cobleigh et al. lament, decomposing a system is tough \cite{Cobleigh:2006}.

In this paper, we propose a new safety verification approach for symbolic specifications that is centered around component selection.
In our approach, we begin by decomposing a system $S$ into components $C_1, \dots, C_n$.
Traditionally, a compositional verification algorithm is applied to these components to verify a system level property $P$, as shown in Fig.~\ref{fig:recomp-visual-1}.
However, verifying these components may be less efficient than verifying the entire (monolithic) system directly without compositional techniques.
%However, this verification problem may be less efficient than verifying the entire (monolithic) system directly without compositional techniques.
Our key insight that addresses this shortcoming is to \textit{recompose} the components into new components $D_1, \dots, D_m$ that we verify instead.
For example, Fig.~\ref{fig:recomp-visual-2} shows $D_1$ composed of $C_1$ and $C_3$ while $D_2$ is composed of $C_2$.

The choice of how to recompose is determined by a \textit{recomposition map} that maps $C_i$'s to $D_j$'s.
Recomposition maps make component selection an explicit, first-class concept and lie at the heart of our technique.
We will show that, in practice, there often exists a recomposition map that results in a compositional verification problem that is more efficient than verifying the monolithic specification directly.

Additionally, we will show that our method is conducive to \textit{specification reduction}.
Specification reduction techniques, e.g. program slicing \cite{Bruckner:2005,Weiser:1982}, are generally considered separately from compositional verification.
However, model checking with recomposition unites these two techniques under a single framework.
For example, Fig.~\ref{fig:recomp-visual-3} shows a situation in which a partial recomposition map is used to reduce a specification with four-components to just the first three.

%Ultimately, finding an efficient verification problem reduces to finding a suitable recomposition map.
Ultimately, selecting components for efficient verification reduces to finding a suitable recomposition map.
Therefore, we propose a technique for automatically selecting recomposition maps.
We use heuristics to prune the large space of possible recomposition maps, which results in a small portfolio of maps that we run in parallel.

\begin{figure}
    %\centering
    \begin{subfigure}[t]{0.99\columnwidth}
        \centering
        \includegraphics[height=2.0cm]{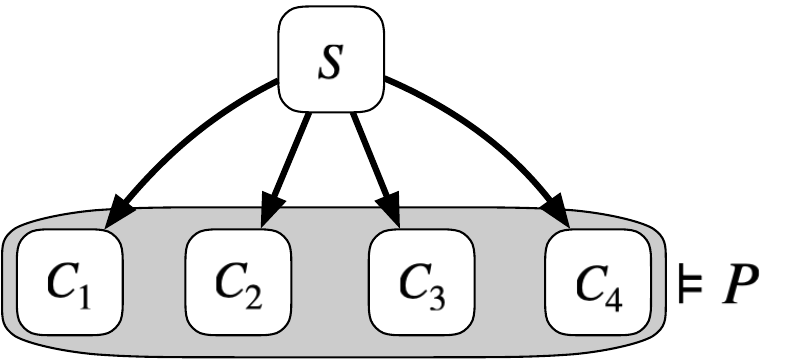}
        \caption{Traditional compositional verification.}
        \label{fig:recomp-visual-1}
    \end{subfigure}
    \begin{subfigure}[t]{0.99\columnwidth}
        \centering
        \includegraphics[height=3.0cm]{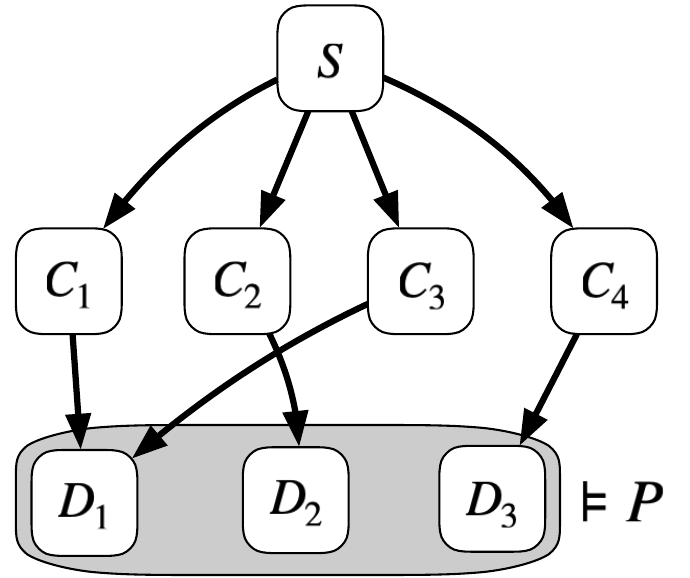}
        \caption{Compositional verification with recomposition.}
        \label{fig:recomp-visual-2}
    \end{subfigure}
    \begin{subfigure}[t]{0.99\columnwidth}
        \centering
        \includegraphics[height=3.0cm]{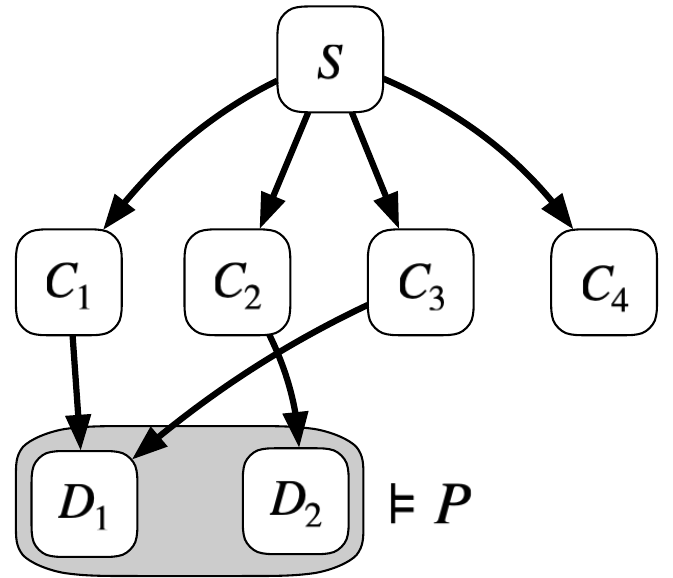}
        \caption{Specification reduction in the recomposition method.}
        \label{fig:recomp-visual-3}
    \end{subfigure}
    \caption{Comparing traditional compositional verification against our recomposition method.}
    \label{fig:recomp-visual}
\end{figure}

We have implemented our techniques in a model checker called ``Recomp-Verify'' for the \TLA{} language \cite{lamport2002specifying}.
In order to bring compositional verification to \TLA{}, we additionally propose a novel parallel composition operator for the language.
We evaluate our techniques by comparing Recomp-Verify to TLC \cite{Yu:1999}, a well-known model checker for \TLA{}.
We show that recomposition can lead to large savings in terms of verification time and the size of the explored state space.

In summary, we make the following contributions:
\textbf{(1) our main contribution, recomposition, which is a technique for efficient compositional verification,}
(2) an automated method for finding efficient recomposition maps using parallelization and heuristics,
(3) a definition for parallel composition for \TLA{} specifications, and
(4) a prototype model checker Recomp-Verify that implements our algorithm, along with an evaluation of Recomp-Verify against TLC on a benchmark of distributed protocols.

%% file: src/Tex/12-running-ex.tex
\section{Motivating Example}
\label{sec:running-ex}
In this section, we describe the Two Phase Commit Protocol \cite{Gray:2006} to motivate our work and serve as a running example throughout the paper.

\subsubsection{Protocol Description}
In the Two Phase Commit Protocol, a \textit{transaction manager} (TM) attempts to commit a transaction onto a pool of \textit{resource managers} (RMs) in two phases.
In the first phase, each RM starts in the \textit{working} state as it attempts to commit the transaction.
Any RM that can commit a transaction sends a \textit{prepared} message to the TM.
In the second phase, if every RM is prepared, the TM will issue a \textit{commit} message to each RM; otherwise the TM will issue an \textit{abort} message.
The protocol assumes that the network can reorder, but not lose, messages.
The key safety property is for each RM to remain \textit{consistent}; i.e. no two RMs should disagree as to whether a transaction was committed or aborted.

\subsubsection{\TLA{} Encoding}
In Fig.~\ref{fig:2pc-tla}, we show only the first (prepare) phase of the $TwoPhase$ specification, which we will refer to as $TP$ for brevity.
$TP$ is a \textit{parameterized} protocol, meaning that the set of RMs in the protocol is given as input.
In the $TP$ specification, the parameter is indicated on line 1 using the keyword \textsc{constant}.

\begin{figure}
    {\tlatex \@x{}\moduleLeftDash\@xx{ {\MODULE} TwoPhase}\moduleRightDash\@xx{}}
    
    \hspace{10pt}
    \begin{subfigure}[t]{0.7\textwidth}
        \internallinenumbers{\small \input{src/TLA/2pc-monolithic}}
    \end{subfigure}
    {\tlatex \@x{}\bottombar\@xx{}}
    \caption{A monolithic encoding of the Two Phase Commit Protocol.}
    \label{fig:2pc-tla}
\end{figure}

$TP$ defines a symbolic transition system (STS) over four state variables, which are declared on line 2 in Fig.~\ref{fig:2pc-tla}.
The variable $rmState$ is the state of each RM, the variables $tmState$ and $tmPrepared$ hold the state of the TM, and $msgs$ is the set of messages each machine sends over the network.
Line 24 formally declares the STS with initial predicate $Init$ and transition relation $Next$.
We show two actions, $SndPrepare$ and $RcvPrepare$, on lines 14 and 9 respectively.
In \TLA{}, actions are typically conjunctions of guards that specify when an action is enabled (lines 10-11 and 15) as well as primed variable expressions that specify transitions (lines 12 and 16-17).
The \textsc{unchanged} keyword on lines 13 and 18 indicate the frame conditions.

The key safety property for the Two Phase Commit Protocol is the invariant $Consistent$.
We can encode this invariant as the following \TLA{} formula:
{\tlatex
\@pvspace{6.0pt}%
%\@x{ \quad Consistent \.{\defeq}}%
%\@x{\@s{16.4} \quad \A\, rm1 ,\, rm2 \.{\in} RMs \.{:} {\lnot} ( rmState [ rm1 ] \.{=}\@w{aborted} \.{\land} rmState
\@x{\@s{4.0} \A\, rm1 ,\, rm2 \.{\in} RMs \.{:}}
\@x{\@s{5.0} {\lnot} ( rmState [ rm1 ] \.{=}\@w{aborted} \.{\land}\, rmState [ rm2 ] \.{=}\@w{committed} )}
%\@x{\@s{7.0} {\lnot} ( rmState [ rm1 ] \.{=}\@w{aborted} \.{\land}\, rmState [ rm2 ] \.{=}\@w{committed} )}
\@pvspace{6.0pt}}%

\subsubsection{Model Checking \textit{TwoPhase}}
The TLC model checker can prove that a given finite instance of $TP$ satisfies the property $Consistent$.
A finite instance of a protocol substitutes a finite value for each parameter, e.g. a finite set of resource managers for $RMs$ in $TP$.
TLC performs explicit state model checking, meaning that it enumerates every possible state in the transition system.
For nine resource managers, TLC is able to prove $TP$ is safe after generating over 10 million states in nearly ten minutes.
However, for ten resource managers, TLC fails to terminate in an hour after checking over 48 million states.
In the following section, we will show how our approach can scale model checking $TP$ to ten resource managers.

\subsubsection{Compositional Verification and Recomposition}
Consider the specifications $RM$, $Env$, $TM_1$, and $TM_2$ shown in Fig.~\ref{fig:2pc-components}.
These specifications represent a decomposition of $TP$; that is, $TP$ is semantically equal to the parallel composition of the four specifications.
We can generate a labeled transition system (LTS) for each of the four specifications in Fig.~\ref{fig:2pc-components} and then use compositional verification techniques to answer the original model checking problem.
For ten resource managers, this strategy enumerates a maximum of 261,002 states and terminates in 1 minute and 32 seconds.

%The compositional verification problem above is more efficient than TLC, but \textit{we can use recomposition to do even better}.
Compositionally verifying the components above is more efficient than TLC, but \textit{we can use recomposition to do even better}.
Later, in example Ex.~\ref{ex:tp}, we use recomposition to identify new components that are \textit{optimal} in terms of minimum run time for verification.
In general, recomposition can provide large savings in terms of run time and state space.
In Sec.~\ref{sec:results}, we show experimentally that recomposition can reduce a model checking problem by \textit{millions} of states.

\begin{figure}
    \begin{multicols}{2}
        {\small \input{src/TLA/2pc-rm-tla}}
        %\vspace{13pt}
        {\small \input{src/TLA/2pc-tm1-tla}}
        \columnbreak
        {\small \input{src/TLA/2pc-env-tla}}
        {\small \input{src/TLA/2pc-tm2-tla}}
    \end{multicols}

    \caption{A decomposition of $TP$.
        %When composed together, the three specifications above equal the monolithic $TP$ encoding.
        Standard operators such as $Spec$, $Next$, $vars$, etc. are omitted for brevity.}
        %Standard operators such as $Spec$, $Next$, $vars$, etc. are defined in the usual way and are omitted for brevity.}
    \label{fig:2pc-components}
\end{figure}
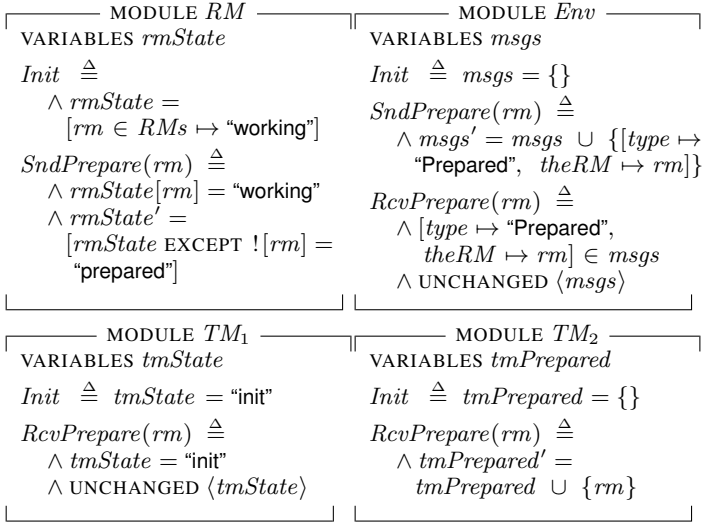

%% file: src/TLA/2pc-monolithic.tex
\tlatex
\@x{ {\CONSTANT} RMs}%
\@x{ {\VARIABLES} msgs ,\, rmState ,\, tmState ,\, tmPrepared}%%
\@pvspace{1.0pt}%
\@x{ vars \.{\defeq} {\langle} msgs ,\, rmState ,\, tmState ,\, tmPrepared {\rangle}}%
%\@x{ RMs \.{\defeq} \{ \@w{rm1} ,\,\@w{rm2} ,\,\@w{rm3} \}}%
\@pvspace{8.0pt}%
\@x{ Init\@s{2.02} \.{\defeq}}%
\@x{\@s{8.2} \.{\land} msgs \.{=} \{ \}}%
 \@x{\@s{8.2} \.{\land} rmState \.{=} [ rm \.{\in} RMs \.{\mapsto}\@w{working}
 ]}%
\@x{\@s{8.2} \.{\land} tmState\@s{0.89} \.{=}\@w{init}}%
\@x{\@s{8.2} \.{\land} tmPrepared \.{=} \{ \}}%
\@pvspace{8.0pt}%
\@x{ RcvPrepare ( rm ) \.{\defeq}}%
 \@x{\@s{8.2} \.{\land} [ type \.{\mapsto}\@w{Prepared} ,\, \ theRM \.{\mapsto}
 rm ] \.{\in} msgs}%
\@x{\@s{8.2} \.{\land} tmState \.{=}\@w{init}}%
\@x{\@s{8.2} \.{\land} tmPrepared \.{'} \.{=} tmPrepared \ \.{\cup} \ \{ rm \}}%
 \@x{\@s{8.2} \.{\land} {\UNCHANGED} {\langle} msgs ,\, tmState ,\, rmState
 {\rangle}}%
\@pvspace{8.0pt}%
\@x{ SndPrepare ( rm ) \.{\defeq}}%
 \@x{\@s{8.2} \.{\land} rmState [ rm ] \.{=}\@w{working}}%
 \@x{\@s{8.2} \.{\land} msgs \.{'} \.{=} msgs \.{\cup} \{ [ type
 \.{\mapsto}\@w{Prepared} ,\, \ theRM \.{\mapsto} rm ] \}}%
 \@x{\@s{8.2} \.{\land} rmState \.{'} \.{=} [ rmState {\EXCEPT} {\bang} [ rm ] \.{=} \@w{prepared} ]}%
 \@x{\@s{8.2} \.{\land} {\UNCHANGED} {\langle} tmState ,\, tmPrepared
 {\rangle}}%
\@pvspace{8.0pt}%
\@x{ Next \.{\defeq}}%
\@x{\@s{16.4} \E\, rm \.{\in} RMs \.{:}}%
\@x{\@s{27.72} \.{\lor} SndPrepare ( rm )}%
\@x{\@s{27.72} \.{\lor} RcvPrepare ( rm )}%
%\@x{\@s{27.72} \.{\lor} SndCommit ( rm )}%
%\@x{\@s{27.72} \.{\lor} RcvCommit ( rm )}%
%\@x{\@s{27.72} \.{\lor} SndAbort ( rm )}%
%\@x{\@s{27.72} \.{\lor} RcvAbort ( rm )}%
%\@x{\@s{27.72} \.{\lor} SilentAbort ( rm )}%
\@pvspace{-3.0pt}%
\hspace{55pt}\vdots
\@pvspace{3.5pt}%
\@x{ Spec \.{\defeq} Init \.{\land} {\Box} [ Next ]_{ vars}}%
\@pvspace{4.0pt}%

%% file: src/TLA/2pc-rm-tla.tex
\tlatex
\@x{}\moduleLeftDash\@xx{ {\MODULE} RM}\moduleRightDash\@xx{}
\@x{ {\VARIABLES} rmState}%%
\@pvspace{4.0pt}%
\@x{ Init\@s{2.02} \.{\defeq}}%
\@x{\@s{8.2} \.{\land} rmState \.{=}}
\@x{\@s{8.2} \.\quad  [ rm \.{\in} RMs \.{\mapsto} \@w{working}]}%
\@pvspace{4.0pt}%
\@x{ SndPrepare ( rm ) \.{\defeq}}%
\@x{\@s{8.2} \.{\land} rmState [ rm ] \.{=}\@w{working}}%
 \@x{\@s{8.2} \.{\land} rmState \.{'} \.{=}}
 \@x{\@s{8.2} \.\quad [ rmState {\EXCEPT} {\bang} [ rm ] \.{=}}
 \@x{\@s{8.2} \ \.\quad \@w{prepared} ]}%
\@pvspace{4.0pt}%
%\@pvspace{14.0pt}%
\@x{}\bottombar\@xx{}

%% file: src/TLA/2pc-tm1-tla.tex
\tlatex
\@x{}\moduleLeftDash\@xx{ {\MODULE} TM_1}\moduleRightDash\@xx{}
\@x{ {\VARIABLES} tmState }%%
%\@pvspace{8.0pt}%
\@pvspace{4.0pt}%
\@x{ Init\@s{2.02} \.{\defeq} tmState\@s{0.89} \.{=}\@w{init}}%
\@pvspace{4.0pt}%
\@x{ RcvPrepare ( rm ) \.{\defeq}}%
\@x{\@s{8.2} \.{\land} tmState \.{=}\@w{init}}%
 \@x{\@s{8.2} \.{\land} {\UNCHANGED} {\langle} tmState
 {\rangle}}%
\@pvspace{2.0pt}%
\@x{}\bottombar\@xx{}

%% file: src/TLA/2pc-env-tla.tex
\tlatex
\@x{}\moduleLeftDash\@xx{ {\MODULE} Env}\moduleRightDash\@xx{}
\@x{ {\VARIABLES} msgs }%%
\@pvspace{4.0pt}%
\@x{ Init\@s{2.02} \.{\defeq} msgs \.{=} \{ \}}%
\@pvspace{4.0pt}%
\@x{ SndPrepare ( rm ) \.{\defeq}}%
 \@x{\@s{8.2} \.{\land} msgs \.{'} \.{=} msgs \ \.{\cup} \ \{ [ type \.{\mapsto}}
 \@x{\@s{8.2} \.\quad \@w{Prepared} ,\, \ theRM \.{\mapsto} rm ] \}}%
\@pvspace{4.0pt}%
\@x{ RcvPrepare ( rm ) \.{\defeq}}%
 \@x{\@s{8.2} \.{\land} [ type \.{\mapsto}\@w{Prepared} ,\,}
 \@x{\@s{8.2} \ \.\quad theRM \.{\mapsto} rm ] \.{\in} msgs}%
 \@x{\@s{8.2} \.{\land} {\UNCHANGED} {\langle} msgs
 {\rangle}}%
\@x{}\bottombar\@xx{}

%% file: src/TLA/2pc-tm2-tla.tex
\tlatex
\@x{}\moduleLeftDash\@xx{ {\MODULE} TM_2}\moduleRightDash\@xx{}
\@x{ {\VARIABLES} tmPrepared}%%
\@pvspace{4.0pt}%
\@x{ Init\@s{2.02} \.{\defeq} tmPrepared \.{=} \{ \}}%
\@pvspace{4.0pt}%
\@x{ RcvPrepare ( rm ) \.{\defeq}}%
\@x{\@s{8.2} \.{\land} tmPrepared \.{'} \.{=}}
\@x{\@s{8.2} \.\quad tmPrepared \ \.{\cup} \ \{ rm \}}%
\@x{}\bottombar\@xx{}

%% file: src/Tex/15-prelims.tex
\section{Preliminaries}
%In this section we formally introduce \emph{labeled transition systems (LTSs)} and a subset of the \TLA{} language that we consider in this paper.
In this section we formally introduce \emph{labeled transition systems (LTSs)}, the \TLA{} language, and the compositional verification technique that we consider in this paper.
Throughout this paper, we will use calligraphic font when referring to LTS variables (e.g. $\mathcal{D}$) and normal font when referring to \TLA{} specifications (e.g. $S$).

\subsection{Labeled Transition Systems}
\label{sec:prelim-lts}
A labeled transition system (LTS) $\mathcal{D}$ is a tuple $(Q,\alpha \mathcal{D}, \delta, I)$ where $Q$ is the set of states, $\alpha \mathcal{D}$ is the alphabet of $\mathcal{D}$, $\delta \subseteq Q \times \alpha \mathcal{D} \times Q$ is the transition relation, and $I$ is a set of initial states.
$\alpha \mathcal{D}$ must be a subset of $\mathbb{A}$, where $\mathbb{A}$ is the universe of all possible actions across all possible LTSs.
%An LTS is \textit{deterministic} if for each pair of transitions $(q,a,q'), (q,a,q'') \in \delta$, we have $q' = q''$.
We let $Reach(\mathcal{D})$ be the set of reachable states in $\mathcal{D}$.
%%%% ad hoc
We define the parallel composition ($\parexp$) over LTSs in the usual way by synchronizing on actions common to both alphabets and interleaving on all other actions \cite{Hoare1978}.

\begin{comment}
We define the parallel composition ($\parexp$) of two LTSs $\mathcal{C}$ and $\mathcal{D}$ as $\mathcal{C} \parexp \mathcal{D} = (Q_\mathcal{C} \times Q_\mathcal{D}, \alpha \mathcal{C} \cup \alpha \mathcal{D}, \delta, I_\mathcal{C} \times I_\mathcal{D})$, where $\delta$ is defined as follows.
Let $t_\mathcal{C} = (q_\mathcal{C},a,q_\mathcal{C}')$ and let $t_\mathcal{D} = (q_\mathcal{D},a,q_\mathcal{D}')$.
Then a transition $((q_\mathcal{C},q_\mathcal{D}),a,(q_\mathcal{C}',q_\mathcal{D}'))$ is in $\delta$ if and only if one of the following holds:
(1) $t_\mathcal{C}$ and $t_\mathcal{D}$ \textit{synchronize}, that is, $a \in \alpha \mathcal{C}$, $a \in \alpha \mathcal{D}$, $t_\mathcal{C} \in \delta_\mathcal{C}$ and $t_\mathcal{D} \in \delta_\mathcal{D}$ or
(2) $t_\mathcal{C}$ and $t_\mathcal{D}$ \textit{interleave}, that is, $a \in \alpha \mathcal{C}$, $a \notin \alpha \mathcal{D}$, $t_\mathcal{C} \in \delta_\mathcal{C}$ and $q_\mathcal{D} = q_\mathcal{D}'$ (or vice versa, with $\mathcal{D}$ and $\mathcal{C}$ swapped).
We will often use the notation $\mathcal{D}_1 \parexp \dots \parexp \mathcal{D}_j$; in the case that $j=0$, we define this expression to be equivalent to $true$, i.e. an LTS that allows all behaviors.
\end{comment}

We define an action-based behavior $\sigma$ as an infinite sequence of actions, i.e. $\sigma \in \mathbb{A}^\omega$, and we let $\sigma_i$ denote the $i$th action in $\sigma$.
We denote the action-based semantics of an LTS $\mathcal{D}$ as a set of action-based behaviors $\aSemantics{\mathcal{D}} \subseteq \mathbb{A}^\omega$.
It is the case that $\sigma \in \aSemantics{\mathcal{D}}$ if and only if there exists a sequence of states $q_0,q_1,\dots \in Q^\omega$ such that $q_0 \in I$ and, for each nonnegative index $i$, either
(1) $\sigma_i \in \alpha \mathcal{D}$ and $(q_i,\sigma_i,q_{i+1}) \in \delta$, or
(2) $\sigma_i \notin \alpha \mathcal{D}$ and $q_i = q_{i+1}$.
Condition (2) allows for \textit{stuttering}, a concept which we will introduce in Sec.~\ref{sec:prelim-tla}.

There are two methods for encoding a safety property as an LTS.
The first method is creating an \textit{error LTS} that includes an error state--which we refer to as the $\pi$ state--that acts as a sink for any action that causes a safety violation.
The second method is creating a \textit{property LTS} whose language defines the safe behaviors; property LTSs must be deterministic and must not include a $\pi$ state.
Any error LTS can be converted to a property LTS using steps two and three for assumption generation (Sec.~3) in \cite{Giann:2002}.
We define property satisfaction over property LTSs as follows: an LTS $\mathcal{D}$ satisfies a property LTS $\mathcal{P}$ ($\mathcal{D} \models \mathcal{P}$) exactly when $\aSemantics{\mathcal{D}} \subseteq \aSemantics{\mathcal{P}}$.
Note that our LTS semantics (with stuttering) properly handles alphabet refinement, and therefore it is unnecessary to consider alphabet restriction \cite{Hoare1978} in our definition of property satisfaction.

\subsection{\TLA{}}
\label{sec:prelim-tla}
In this paper we will refer to a \TLA{} specification $S$ as a syntactic entity that consists of constants, variables and operator definitions, etc. in the format shown in Fig.~\ref{fig:2pc-tla}.
%A specification need not have constants, but each specification must declare at least one state variable.
%Specifications also define a sequence $vars$ that contains every state variable.

The initial state predicate, transition relation, and specification declaration are named $Init$, $Next$, and $Spec$ respectively.
In this paper, $Init$, $Next$, and $Spec$ are restricted to the syntax of $init$, $next$, and $spec$ given by the grammar in Fig.~\ref{fig:tla-grammar}.
In $next$, the domain $D$ does not contain state variables.
We also restrict action definitions to the syntax of $op$, and no actions are referenced in the body of another action.
In the grammar, $\Box$ is the \textit{always} temporal operator.
Expression $[Next]_{vars}$ is equal to $Next \lor (vars' = vars)$ and allows for \textit{stuttering} states, i.e. consecutive states whose variables in $vars$ do not change.

\begin{figure}
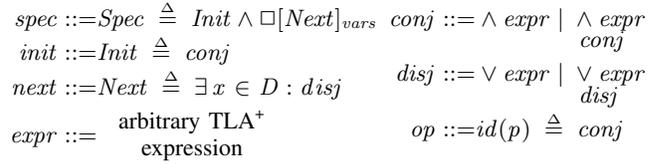

    \centering
    \begin{multicols}{2}
        \tlatex
        \small
        \begin{align*}
            spec ::= &Spec \.{\defeq} Init \.{\land} \Box[Next]_{vars}\\
            init ::= &Init \.{\defeq} conj\\
            next ::= &Next \.{\defeq} \E\, x \.{\in} D \.{:} \.disj\\
            expr ::= &\ \sstack{\text{arbitrary \TLA{}} \\ \text{expression}}
        \end{align*}
        \begin{align*}
            conj ::= &\.{\land} expr \mid \.{\land} expr\\[-5pt]
            &\hspace{39pt} conj\\
            disj ::= &\.{\lor} expr \mid \.{\lor} expr\\[-5pt]
            &\hspace{39pt} disj\\
            op ::= &id(p) \.{\defeq} conj
        \end{align*}
    \end{multicols}
    %\vspace{-15pt}
    \caption{Restricted \TLA{} grammar for this paper.}
    \label{fig:tla-grammar}
    %\vspace{-5pt}
\end{figure}

We define several operators over a \TLA{} specification $S$.
The scoping operator $!$ references definitions in $S$, e.g. $TP!SndPrepare$ refers to the $SndPrepare$ action of $TP$ in Fig.~\ref{fig:2pc-tla}.
The operators $\salpha$ and $\calpha$ denote \textit{symbolic actions} and \textit{concrete actions} respectively.
Symbolic actions are the action names in a specification, while concrete actions are the actions that may occur in a finite instance.
For example, let $TP_1$ be the finite instance of $TP$ with $RM = \{$``$rm1"\}$, then $\salpha TP_1 = \salpha TP = \{SndPrepare, RcvPrepare\}$ and $\calpha TP_1 = \{SndPrepare($``$rm1"), RcvPrepare($``$rm1")\}$.
Additionally, we let $\vbeta S$ denote the set of state variables in a specification or an expression, e.g. $\vbeta TP = \{msgs,rmState,tmState,tmPrepared\}$.
For an operator $* \in \{\salpha,\calpha,\vbeta\}$ and a set of specifications $Z$, the notation $* Z$ is short-hand for the union of $* z$, for each specification $z \in Z$.
%The operators $\salpha$, $\calpha$, and $\vbeta$ may also be applied to a set of specifications, which is short-hand for the union of the operation application to each specification in the set.
%The operators $\salpha$, $\calpha$, and $\vbeta$ may also be applied to a set of specifications, which is short-hand for the union of applying the operation to each specification in the set.

To define the semantics of a \TLA{} formula, we first define a state as an assignment to all state variables.
Then, the semantics of a \TLA{} formula is a set of behaviors, where a behavior is an infinite sequence of states.
We indicate state-based semantics of a \TLA{} formula $F$ as $\bSemantics{F}$, the set of behaviors that satisfy $F$.
For a \TLA{} specification $S$, we will often abbreviate $\bSemantics{S!Spec}$ to simply $\bSemantics{S}$.
Given a \TLA{} property $P$, we say $S$ satisfies $P$ ($S \models P$) exactly when $\bSemantics{S} \subseteq \bSemantics{P}$.

We define the operator $\tolts(S)$, which converts a \TLA{} specification $S$ into an LTS $\mathcal{D}$.
$\tolts(S)$ can be realized by generating the full state graph for $S$ and then labeling its edges with the concrete actions $\calpha S$ such that $\alpha \mathcal{D} = \calpha S$.
Additionally, we define two operators for converting \TLA{} properties to an LTS.
The first operator, $\errlts(S,P)$, constructs an error LTS for $S$ where violations of $P$ lead to a $\pi$ state.
The second operator, $\proplts(S,P)$, builds a property LTS for $S$ where no violation of $P$ is possible.
$\proplts(S,P)$ can be constructed from $\errlts(S,P)$, as pointed out in Sec.~\ref{sec:prelim-lts}.

\subsection{CRA-Style Compositional Verification}
%There are several different flavors of compositional verification \cite{Giannakopoulou:2018handbook}, including \textit{assume-guarantee} \cite{Jones:1983,Pnueli:1985} and \textit{compositional reachability analysis} (CRA) \cite{Cheung:1995,Graf:1990,Yeh:1991}.
%While assume-guarantee techniques have been well-studied in the context of model checking, only limited success has been reported for verifying multiple components \cite{Nam:2008,Pasareanu:2006}.
%We require a compositional verification algorithm that works for multiple components, so we choose to use CRA-style techniques, which have reported success for verifying safety properties for multiple components \cite{Cheung:1999}.

In this paper, we consider a style of compositional verification called \textit{compositional reachability analysis} (CRA) \cite{Cheung:1995,Graf:1990,Yeh:1991}.
Our recomposition framework requires a compositional verification algorithm that works for multiple components, and CRA-style techniques have reported success for verifying safety properties of multi-component systems \cite{Cheung:1999}.

CRA is used to check safety by composing the LTS for each component together in a hierarchical fashion; safety is proved if and only if the $\pi$ state is unreachable in the overall system.
Such algorithms generally derive their divide-and-conquer efficiency from two optimizations: intermediate minimization and short-circuiting.
The former involves minimizing the state space of the intermediate LTSs with respect to observational equivalence \cite{milner:1989ccs} during composition.
The latter optimization, short-circuiting, occurs when a strict subset of components are needed for verification to succeed.
In this case, the remaining components (outside the strict subset) can be skipped, and hence short-circuiting provides a \textit{dynamic} form of specification reduction.

%% file: src/Tex/17-parallel-comp.tex
\section{Parallel Composition in \TLA{}}
\label{sec:parallel-comp}

%\ek{It might be helpful to briefly mention why you are introducing a new parallel composition operator for TLA+; otherwise, the reader might feel that it comes out of nowhere}
%In this section, we introduce a parallel composition operator over the \TLA{} language that will be central to our recomposition algorithm.
%Throughout this paper, we will use two notions of parallel composition.
%The second is a syntactic notion of parallel composition over \TLA{} specifications; the operator is syntactic because it is defined entirely based on \TLA{} syntax, and does not involve explicitly enumerating the state space.

In this section, we introduce a new parallel composition operator over \TLA{} specifications.
The operator is central to our recomposition algorithm and will allow us to define concepts such as decomposition and recomposition in Sec.~\ref{sec:algorithm}.
The new operator is syntactic; in other words, the definition is entirely in terms of \TLA{} syntax, and does not involve explicitly enumerating the state space.
To avoid confusion between the parallel composition operator $\parexp$ over LTSs (Sec.~\ref{sec:prelim-lts}), we will denote the \TLA{} parallel composition operator using $\parsym$.
We will use the notation $\parsym Z$ to denote the composition over a set of specifications $Z$.
We now define $\parsym$ in the usual way, by synchronizing common actions between specifications and interleaving all others actions \cite{Hoare1978}.
\begin{definition}[Parallel Composition]
    \label{def:parallel-comp}
    Let $S$ and $T$ be \TLA{} specifications with distinct state variables. We define $S \parsym T$ as follows.
    First, $S \parsym T$ contains exactly the constants and state variables in $S$ and $T$.
    Second, in $S \parsym T$, we define $vars \triangleq (S!vars) \circ (T!vars)$, where $\circ$ is the sequence concatenation operator.
    Finally, in $S \parsym T$, we define $Spec \triangleq Init \land [Next]_{vars}$ where $Init \triangleq S!\text{Init} \land T!\text{Init}$ and $Next$ is defined as follows:
    \begin{align*}
        &\bigvee_{A \in \aalpha S \cup \aalpha T} \left\{ \tlastate{
            \exists d \in D \st S!A(d) \ \land \ T!A(d) \\
            \quad \text{if} \ A \in \aalpha S \ \text{and} \ A \in \aalpha T \\
            \exists d \in D \st S!A(d) \ \land \ T!vars' = T!vars \\
            \quad \text{if} \ A \in \aalpha S \ \text{and} \ A \notin \aalpha T \\
            \exists d \in D \st T!A(d) \ \land \ S!vars' = S!vars \\
            \quad \text{if} \ A \notin \aalpha S \ \text{and} \ A \in \aalpha T
        } \right.
    \end{align*}
\end{definition}
%\begin{remark}
    %Let $Z$ be a set of \TLA{} specifications, then we use the notation $\parsym Z$ to denote the result of composing every specification in $Z$.
    %We will use this notation in Alg.~\ref{alg:recomp-verify}.
%\end{remark}
Notice that Def.~\ref{def:parallel-comp} defines parallel composition in terms of \TLA{} syntax, and hence does not increase the expressivity of the language.
While the operator itself is novel, this technique is briefly discussed by Lamport \cite{lamport2002specifying}.
\begin{example}
    By Def.~\ref{def:parallel-comp}, $TP = RM \parsym Env \parsym TM_1 \parsym TM_2$.
    Furthermore, consider specifications $T_1$ and $T_2$ from Fig.~\ref{fig:2pc-t1-t2}.
    Notice that $TP = RM \parsym T_1$, $T_1 = Env \parsym T_2$, and $T_2 = TM_1 \parsym TM_2$.
\end{example}
\begin{figure}
    \begin{multicols}{2}
        {\small \input{src/TLA/2pc-T1-tla}}
        %\columnbreak
        {\small \input{src/TLA/2pc-T2-tla}}
    \end{multicols}

    \caption{Intermediate decomposed specifications $T_1$ and $T_2$ in the $TP$ example.}
    \label{fig:2pc-t1-t2}
\end{figure}
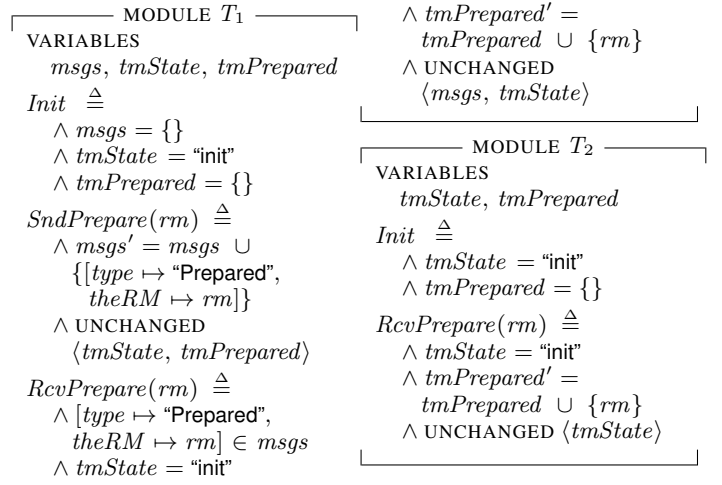

The following theorem shows that parallel composition behaves exactly as we expect if we convert a \TLA{} specification to an LTS.
We prove this theorem using more general semantics for \TLA{} specifications, namely action-state-based semantics.
We include a proof in Appendix~\ref{apx:composition-sem-proof}.
\begin{theorem}
    \label{thm:composition-sem}
    $\aSemantics{\tolts(S \parsym T)} = \aSemantics{\tolts(S) \parexp \tolts(T)}$.
\end{theorem}

%% file: src/TLA/2pc-T1-tla.tex
\tlatex
\@x{}\moduleLeftDash\@xx{ {\MODULE} T_1}\moduleRightDash\@xx{}
\@x{ {\VARIABLES}}%%
\@x{ \quad msgs ,\, tmState ,\, tmPrepared}%%
\@pvspace{4.0pt}%
\@x{ Init\@s{2.02} \.{\defeq}}%
\@x{\@s{8.2} \.{\land} msgs \.{=} \{ \}}%
\@x{\@s{8.2} \.{\land} tmState\@s{0.89} \.{=}\@w{init}}%
\@x{\@s{8.2} \.{\land} tmPrepared \.{=} \{ \}}%
\@pvspace{4.0pt}%
\@x{ SndPrepare ( rm ) \.{\defeq}}%
 \@x{\@s{8.2} \.{\land} msgs \.{'} \.{=} msgs \ \.{\cup}}
 \@x{\@s{8.2} \quad \{ [ type \.{\mapsto}\@w{Prepared} ,\,}
 \@x{\@s{8.2} \.\quad \ \ theRM \.{\mapsto} rm ] \}}%
 \@x{\@s{8.2} \.{\land} {\UNCHANGED}}
 \@x{\@s{8.2} \quad {\langle} tmState ,\, tmPrepared {\rangle}}%
\@pvspace{4.0pt}%
\@x{ RcvPrepare ( rm ) \.{\defeq}}%
 \@x{\@s{8.2} \.{\land} [ type \.{\mapsto}\@w{Prepared} ,\,}
 \@x{\@s{8.2} \quad theRM \.{\mapsto} rm ] \.{\in} msgs}
\@x{\@s{8.2} \.{\land} tmState \.{=}\@w{init}}%
\@x{\@s{8.2} \.{\land} tmPrepared \.{'} \.{=}}
\@x{\@s{8.2} \quad tmPrepared \ \.{\cup} \ \{ rm \}}%
 \@x{\@s{8.2} \.{\land} {\UNCHANGED}}
 \@x{\@s{8.2} \quad {\langle} msgs ,\, tmState {\rangle}}%
\@x{}\bottombar\@xx{}

%% file: src/TLA/2pc-T2-tla.tex
\tlatex
\@x{}\moduleLeftDash\@xx{ {\MODULE} T_2}\moduleRightDash\@xx{}
\@x{ {\VARIABLES}}
\@x{ \quad tmState ,\, tmPrepared}%%
\@pvspace{4.0pt}%
\@x{ Init\@s{2.02} \.{\defeq}}%
\@x{\@s{8.2} \.{\land} tmState\@s{0.89} \.{=}\@w{init}}%
\@x{\@s{8.2} \.{\land} tmPrepared \.{=} \{ \}}%
\@pvspace{4.0pt}%
\@x{ RcvPrepare ( rm ) \.{\defeq}}%
\@x{\@s{8.2} \.{\land} tmState \.{=}\@w{init}}%
\@x{\@s{8.2} \.{\land} tmPrepared \.{'} \.{=}}
\@x{\@s{8.2} \quad tmPrepared \ \.{\cup} \ \{ rm \}}
 \@x{\@s{8.2} \.{\land} {\UNCHANGED} {\langle} tmState {\rangle}}%
\@x{}\bottombar\@xx{}

%% file: src/Tex/20-overview.tex
\section{Model Checking with Recomposition}
\label{sec:algorithm}
In this section, we propose our algorithm for verifying symbolic specifications.
We begin by introducing the algorithm in Sec.~\ref{sec:recomp-verify-alg}. %, followed by the details of decomposition in Sec.~\ref{sec:sym-decomposition} and the details for compositional verification in Sec.~\ref{sec:comp-verif}.
Subsequently, we provide details for decomposition (Sec.~\ref{sec:sym-decomposition}), static specification reduction (Sec.~\ref{sec:static-spec-reduction}), and compositional verification (Sec.~\ref{sec:comp-verif}).
Finally, we conclude this section with a correctness analysis of the algorithm in Sec.~\ref{sec:correctness-analysis}.

\subsection{The Recomp-Verify Algorithm}
\label{sec:recomp-verify-alg}

\subsubsection{Algorithm Overview}
\label{sec:recomp-verify-overview}
Our algorithm solves a model checking problem $S \models P$, where $S$ and $P$ are both written in \TLA{}.
The algorithm begins by decomposing $S$ into $n$ components $C_1, \dots, C_n$, each of which is also a \TLA{} specification.
The decomposition algorithm ensures two key properties upon termination:
(P1) $S = C_1 \parsym \dots \parsym C_n$ and
(P2) the first component, $C_1$, contains all state variables that occur syntactically in $P$.
Property (P1) ensures soundness of the decomposition, while property (P2) allows us to build the safety property $\prop$ described in the following paragraph.

After decomposition, the algorithm \textit{recomposes} the $C_i$ components into new components $D_P$ and $D_1, \dots, D_m$.
These new components define the following compositional verification problem that is equivalent to the original: $\tolts(D_1) \parexp \dots \parexp \tolts(D_m) \models \prop$, where $\prop = \proplts(D_P,P)$.
For $\proplts(D_P,P)$ to be well-formed, $D_P$ must contain every state variable that occurs in $P$.
Therefore, we require $D_P$ to be composed of (at least) $C_1$, as $C_1$ must contain every state variable that occurs in $P$ by property (P2) of decomposition.
We formally capture this requirement, as well as the choice of how to perform recomposition, in the following definition.
\begin{definition}[Recomposition Map]
    \label{def:recomp-map}
    A \textit{recomposition map} is a surjective function $f : \{C_1,\dots,C_n\} \rightarrow \{d_P,d_1,\dots,d_m\}$ such that $f(C_1) = d_P$.
\end{definition}
In Def.~\ref{def:recomp-map}, the $d_j$'s in the co-domain are intended as a placeholder for constructing each $D_j$.
In particular, we will define each recomposed component as $D_j = \parsym f^{-1}(d_j)$, the parallel composition of one or more $C_i$ components.
Therefore, the restriction $f(C_1) = d_P$ implies that $D_P$ will be composed of $C_1$ as intended.
Finally, once each $D_j$ is constructed, we solve the compositional verification problem.

\subsubsection{Algorithm Details}
\label{sec:recomp-verify}
We present our model checking algorithm in Alg.~\ref{alg:recomp-verify}.
The algorithm accepts several inputs, including a \textit{recomposition strategy}.
A recomposition strategy $\rho$ is a function that maps $C_i$ components to a pair $(f,m)$, where $f$ is a total recomposition map and $m$ is the number of $D_j$ components.
In other words, the recomposition strategy determines which recomposition map is used.
%We require strategies to produce a total recomposition map to ensure verification is sound; however, in Sec.~\ref{sec:static-spec-reduction}, we detail how static specification reduction may alter $f$ to be partial while maintaining soundness.
In the remainder of this section we assume $\rho$ is given; we discuss recomposition strategy selection in Sec.~\ref{sec:parallel-alg}.

\begin{algorithm}
\centering
\caption{\textsc{Recomp-Verify}}
\label{alg:recomp-verify}
\renewcommand{\algorithmicrequire}{\textbf{Input:}}
\renewcommand{\algorithmicensure}{\textbf{Output:}}
\begin{algorithmic}[1]
\Require{Specification $S$, property $P$, recomposition strategy $\rho$}
\Ensure{If $S \models P$}
\State $C_1, \dots, C_n = \textsc{Decompose}(S,P)$
%\State $f,m \leftarrow \textsc{Static-Reduce}(\rho(C_1, \dots, C_n))$
\State $f,m \leftarrow \rho(C_1, \dots, C_n)$
\State $f,m \leftarrow \textsc{Static-Reduce}(f,m)$
\State $D_P \leftarrow \parsym f^{-1}(d_P)$ \Comment{$f^{-1}(d_P) \subseteq \{C_1, \dots, C_n\}$}
\For{$j \in \{1,\dots,m\}$}
    \State $D_j \leftarrow \parsym f^{-1}(d_j)$ \Comment{$f^{-1}(d_j) \subseteq \{C_2, \dots, C_n\}$}
\EndFor
\State \Return $\textsc{Comp-Verify}(D_1,\dots,D_m,D_P,P)$
\end{algorithmic}
\end{algorithm}

Alg.~\ref{alg:recomp-verify} begins by decomposing $S$ into components on line 1. %; we provide more detail for decomposition in Sec.~\ref{sec:sym-decomposition}.
The strategy $\rho$ selects a recomposition map on line 2, which is possibly statically reduced on line 3. %; we provide more detail for static specification reduction in Sec.~\ref{sec:static-spec-reduction}.
We provide more detail for decomposition and static specification reduction in Sec.~\ref{sec:sym-decomposition} and Sec.~\ref{sec:static-spec-reduction} respectively.
Next, on lines 4-6, we perform recomposition using the recomposition map $f$.
On line 4, we define $D_P$ to be the parallel composition of each $C_i$ component in the pre-image $f^{-1}(d_P)$.
Similarly, on line 6, we define each $D_j$ to be the parallel composition of each $C_i$ component in the pre-image $f^{-1}(d_j)$.
Finally, on line 7, we solve the compositional verification problem for the recomposed components ($D_j$'s); we provide more detail for this step in Sec.~\ref{sec:comp-verif}.

\begin{example}
    \label{ex:tp}
    In this example we analyze Alg.~\ref{alg:recomp-verify} given the input $TP$, $Consistent$, and a hand-crafted optimal recomposition strategy $\rho_{opt}$.
    Line 1 of Alg.~\ref{alg:recomp-verify} produces the components $RM$, $Env$, $TM_1$, and $TM_2$ from Fig.~\ref{fig:2pc-components}.
    On line 2, $\rho_{opt}$ chooses $m=2$ and $f$ such that $f(RM) = d_P$, $f(Env) = f(TM_1) = d_1$, and $f(TM_2) = d_2$.
    Static specification reduction on line 3 has no effect on $f$ and $m$.
    Recomposition (lines 4-6) reduces the original model checking problem to $\tolts(Env \parsym TM_1) \parexp \tolts(TM_2) \models \proplts(RM,Consistent)$, which we solve on line 7.
    Whereas the example in Sec.~\ref{sec:running-ex} verifies four specifications (for $RM$, $Env$, $TM_1$, $TM_2$), this example verifies three (for $RM$, $Env \parsym TM_1$, $TM_2$).
    The strategy $\rho_{opt}$ in this example reduces the maximum state space by 1,027 states and improves the model checking time from 1 minute 32 seconds to 51 seconds.
\end{example}

\begin{comment}
\begin{figure}
    %\vspace{-10pt}
    \centering
    \begin{subfigure}[t]{0.45\textwidth}
        \centering
        \includegraphics[height=3.0cm]{recomp-verify FMCAD24/Figs/tpm-naive-recomp.pdf}
        \caption{Naive recomposition map.}
        \label{fig:naive-recomp-map}
    \end{subfigure}
    \hfill %\hspace{31pt}%
    \begin{subfigure}[t]{0.45\textwidth}
        \centering
        \includegraphics[height=3.0cm]{recomp-verify FMCAD24/Figs/tpm-optimal-recomp.pdf}
        \caption{Optimal recomposition map.}
        \label{fig:optimal-recomp-map}
    \end{subfigure}
    \caption{We show two possible recomposition maps for $TP$, the naive choice (a) and the optimal choice (b).}
    \label{fig:recomp-maps}
    %\vspace{-10pt}
\end{figure}
\end{comment}

%% file: src/Tex/25-decomp-alg.tex
\subsection{Decomposition}
\label{sec:sym-decomposition}
In this section, we present an algorithm for decomposing a symbolic specification $S$ into $n$ components $C_1 \dots C_n$.
Our algorithm guarantees the following two properties:
(P1) $S = C_1 \parsym \dots \parsym C_n$ and
(P2) $\vbeta P \subseteq \vbeta C_1$.
We provide a correctness argument for these two properties in Sec.~\ref{sec:correctness-analysis}.

\subsubsection{Decomposition Algorithm}
Each step of the algorithm splits a specification $T_i$ into two specifications $C_{i+1}$ and $T_{i+1}$ such that $T_i = C_{i+1} \parsym T_{i+1}$.
We note the following two corner cases: $T_0 = S$ and $C_n = T_{n-1}$.
The algorithm splits a specification across two phases: \textit{state variable partitioning} and \textit{specification slicing}.
The former partitions the state variables of $T_i$ into two sets $V_C$ and $V_T$, while the latter slices $T_i$ into $C_{i+1}$ and $T_{i+1}$ that contain the variables $V_C$ and $V_T$ respectively.
We present the algorithm in Alg.~\ref{alg:decompose}.
We now explain state variable partitioning and specification slicing in detail across the following two sections.

\begin{example}
    \label{ex:decomp-alg}
    We explain Alg.~\ref{alg:decompose} given $TP$ and $Consistent$.
    The algorithm begins with the partition $V_C = \{rmState\}$ and $V_T = \{msgs,rmState,rmPrepared\}$ on line 1; we explain partitioning in Sec.~\ref{sec:state-var-partitioning}.
    Next, on lines 6-7, the algorithm slices $TP$ into $RM$ (Fig.~\ref{fig:2pc-components}) and $T_1$ (Fig.~\ref{fig:2pc-t1-t2}).
    The state variables of $T_1$ are subsequently partitioned into $V_C = \{msgs\}$ and $V_T = \{rmState,rmPrepared\}$ on line 9.
    The algorithm continues in this fashion until $V_T = \emptyset$, i.e. no partition is possible.
    The algorithm will then exit the loop and return the components $RM$, $Env$, $TM_1$, $TM_2$ on line 12.
\end{example}

%%%%%%%%%%%%%%%%%%%%%%%%%%%%%%%%%%%%%%%%%%%%%%%%%%%%%%%%%%%%%%%%%%%%%%%%%%%%%%%%%%
\begin{algorithm}[H]
\centering
\caption{\textsc{Decompose}}
\label{alg:decompose}
\renewcommand{\algorithmicrequire}{\textbf{Input:}}
\renewcommand{\algorithmicensure}{\textbf{Output:}}
\begin{algorithmic}[1]
    \Require{Specification $S$, Safety Property $P$}
    \Ensure{$C_1, \dots, C_n$ with properties (P1) and (P2)}
    \State $V_C, V_T \leftarrow \textsc{Partition}(S,\vbeta P)$
    \If{$V_T = \emptyset$}
        \State \Return $S$
    \EndIf
    \State $T_0 \leftarrow S$, $i \leftarrow 0$
    \While{$V_T \neq \emptyset$}
        \State $C_{i+1} \leftarrow \textsc{Slice}(T_i, V_C)$
        \State $T_{i+1} \leftarrow \textsc{Slice}(T_i, V_T)$
        \State $v \in \vbeta T_{i+1}$ \Comment{Nondeterministically choose a variable}
        \State $V_C, V_T \leftarrow \textsc{Partition}(T_{i+1},\{v\})$
        \State $i \leftarrow i+1$
    \EndWhile
    \State $n \leftarrow i+1$, $C_n \leftarrow T_i$
    \State \Return $C_1, \dots, C_n$
    \Procedure{Partition}{$T,V$}
    \State $V_C \leftarrow \fixpoint(\textsc{Occurs}_T,V)$
    \State $V_T \leftarrow \vbeta S - V_C$
    \State \Return $V_C, V_T$
    \EndProcedure
    \Procedure{$\text{Occurs}_S$}{$V$} 
    \State \Return $\bigcup\limits_{A \in \salpha S} \{ \vbeta c \mid c \in \text{Conj}(A) \text{ and } \vbeta c \cap V \neq \emptyset \}$
    %\Statex \hspace{10px} $\bigcup\limits_{A \in \salpha S} \{ \vbeta c \mid c \in \text{Conj}(A) \text{ and }$
    %{\tlatex \@pvspace{-6.0pt}}
    %\Statex \hspace{54pt} $\vbeta c \cap V \neq \emptyset \}$
    \EndProcedure
    \Procedure{$\fixpoint$}{$op,X$}
    \State $Y \leftarrow X \cup op(X)$
    \If{X = Y} \Return $X$
    \EndIf
    \State \Return $Y \cup \fixpoint(op,Y)$
    \EndProcedure
\end{algorithmic}
\end{algorithm}
%%%%%%%%%%%%%%%%%%%%%%%%%%%%%%%%%%%%%%%%%%%%%%%%%%%%%%%%%%%%%%%%%%%%%%%%%%%%%%%%%%

\subsubsection{State Variable Partitioning}
\label{sec:state-var-partitioning}
Given a specification $T_i$, the partitioning phase partitions the variables $\vbeta T_i$ into two sets $V_C$ and $V_T$.
The partition procedure appears twice in Alg.~\ref{alg:decompose}.
The first occurrence, on line 1, determines the state variables that will appear in $C_1$; therefore, to uphold property (P2), we partition on $\vbeta P$.
In the second appearance, on line 9, we choose just \textit{one} variable in attempt to produce as many components as possible (ideally, one component per state variable).
We are free to choose the one variable nondeterministically because the order of decomposed components is inconsequential; this is due to the fact that the ordering is ultimately determined by a recomposition map in Alg.~\ref{alg:recomp-verify}.

The partition procedure in Alg.~\ref{alg:decompose} also guarantees that the state variables in each partition will constitute a well-formed slice according to the grammar in Fig.~\ref{fig:tla-grammar}.
For example, if a specification contains the expression $a = b + 1$, then $a$ and $b$ should be grouped together into the same partition.
To accomplish this, we let $V_C$ be $V$ \textit{plus} any variables that occur within the same expression, repeated until fix-point.
More formally, we let $V_C = \fixpoint(\textsc{Occurs}_S,V)$ (line 14), where $\fixpoint$ invokes the $\textsc{Occurs}_S$ procedure, initially on $V$, until a fix-point is reached.
Finally, we choose $V_T$ to be the remainder of the state variables in $T_i$ (line 15).
\begin{example}
    Notice that $\vbeta Consistent = \{rmState\}$ and $\fixpoint(\textsc{Occurs}_{TP},$ $\{rmState\}) = \{rmState\}$.
    Therefore, the first partition (line 1) will be $V_C = \{rmState\}$ and $V_T = \{msgs,tmState,tmPrepared\}$.
    In the second partition (lines 8-9), we arbitrarily choose $v = msgs$, which results in $V_C = \{msgs\}$ and $V_T = \{tmState,tmPrepared\}$.
\end{example}

\subsubsection{Specification Slicing}
The specification slicing phase restricts a specification $T_i$ to a given subset of its variables $V$.
Slicing can be seen as the inverse of parallel composition.
For example, consider a system specification $M$ with action $Action$ and state variables $var_1$ and $var_2$:
\vspace{5pt}
{\tlatex
\@x{\@s{68} Action \.{\defeq} \.{\land} var_1 \.{'} \.{=} \@w{val1}}%
 \@x{\@s{68} \@s{48.8} \.{\land} var_2 \.{'} \.{=} \@w{val2}}
 %\@x{\@s{76} \@s{39.9} \.{\land} var_2 \.{'} \.{=} \@w{val2}}
}\\[5pt]
Given the variable partition $\{var_1\}$, $\{var_2\}$, we can view $M$ as the composition of two components $M_1$ and $M_2$ that respectively define: ${\tlatex Action \.{\defeq} var_1 \.{'} \.{=} \@w{val1}}$ and ${\tlatex Action \.{\defeq} var_2 \.{'} \.{=} \@w{val2}}$.
In particular, we have $M_1 = \textsc{Slice}(M,\{var_1\})$, $M_2 = \textsc{Slice}(M,\{var_2\})$, and $M = M_1 \parsym M_2$.
In the $TP$ example, this corresponds to $TP = RM \parsym T_1$ in Ex.~\ref{ex:decomp-alg}.
We include more details on slicing, including the definition for the slicing procedure, in Appendix~\ref{apx:slicing-procedure}.

%% file: src/Tex/30-static-reduction.tex
\subsection{Static Specification Reduction}
\label{sec:static-spec-reduction}

In Sec.~\ref{sec:recomp-verify-alg}, we require recomposition strategies to produce a total recomposition map.
Total recomposition maps apply verification to every component; however, in some cases, not every component is necessary for verification.
Therefore, in the following paragraph, we introduce a technique for statically detecting a subset of components that are necessary for verification.
In Alg.~\ref{alg:recomp-verify}, the procedure $\textsc{Static-Reduce}(f,m)$ on line 3 restricts the domain of $f$ to this subset and reduces the codomain and $m$ accordingly so $f$ remains surjective.

\begin{comment}
However, in the following paragraph, we show how to statically compute a subset of the components that are necessary for verification.
%However, we are able to statically compute a subset of the components that are necessary for verification.
The procedure $\textsc{Static-Reduce}(f,m)$ restricts the domain of $f$ to the components that are necessary for verification, and may cause $f$ to be partial.
%Therefore, the procedure $\textsc{Static-Reduce}(f,m)$ may alter a recomposition map $f$ to be partial, defined only for the components that are necessary for verification.
This procedure is a \textit{static specification reduction} technique, akin to program slicing \cite{Bruckner:2005,Weiser:1982}.
%We therefore allow recomposition maps to be partial, defined only for the components that are necessary for verification.
%Partial recomposition maps are a \textit{static specification reduction} technique, akin to program slicing \cite{Bruckner:2005,Weiser:1982}.
\end{comment}

The subset of necessary components is those whose alphabets may affect--either directly or indirectly--the actions of $C_1$, and therefore may prevent the entire system from reaching an error.
More formally, the subset of components is $\bigcup_i X_i$, where $X_0 = \{C_1\}$ and $X_{i+1} = \{ C_j \mid \salpha C_j \cap \salpha X_i \neq \emptyset \}$.
In Appendix~\ref{apx:static-reduction}, we show that it is only necessary to consider the first $n+1$ terms--where $n$ is the number of $C_i$ components--when computing the union of the $X_i$'s.
Intuitively, $X_1$ is the set of components that may directly prevent $C_1$ from reaching an error, while $X_2$, $X_3$, etc. may indirectly prevent an error.
\begin{example}
    \label{ex:tp-counter}
    We now introduce $TPCounter$, an extension to $TP$.
    $TPCounter$ is identical to $TP$, except it includes one more state variable $counter$ and one more action $Increment$.
    In the initial state, $counter$ is equal to zero.
    Each original action from $TP$ leaves $counter$ unchanged, while $Increment$ increments $counter$ by one and leaves all other state variables unchanged.
    The $Increment$ action is always enabled, and therefore $TPCounter$ is an infinite-state protocol.

    Consider model checking $TPCounter \models Consistent$ with Alg.~\ref{alg:recomp-verify}.
    Decomposition (line 1) produces five components: $RM$, $Env$, $TM_1$, $TM_2$, and $Counter$, where $Counter$ has one state variable $counter$ and one action $Increment$.
    $Counter$ is the only specification with the action $Increment$ and, therefore, does not synchronize with the actions in the other four specifications.
    Therefore, $Counter$ cannot affect the safety of $C_1$.
    Formally, $X_0 = \{RM\}$, $X_1 = \{RM,Env\}$, $X_2 = \{RM,Env,TM_1,TM_2\}$, $X_3 = X_2$, etc. so $Counter$ is not a necessary component.
    \textsc{Static-Reduce} will therefore omit $Counter$ from the domain of any given recomposition map, causing Alg.~\ref{alg:recomp-verify} to successfully terminate.
\end{example}

\begin{comment}
Thm.~\ref{thm:static-reduction} formally shows that it is sound and complete to reduce verification to the necessary components; we prove this theorem in Appendix~\ref{apx:static-reduction}.
\begin{theorem}
    \label{thm:static-reduction}
    $S \models P$ if and only if $\parsym \left( \bigcup_i X_i \right) \models P$.
\end{theorem}
\end{comment}

%% file: src/Tex/40-comp-verif.tex
\subsection{Compositional Verification}
\label{sec:comp-verif}
We present a CRA-style compositional verification algorithm in Alg.~\ref{alg:comp-verify}.
The algorithm works by iteratively composing the LTS for each component $D_j$ together (line 5) until the $\pi$ state becomes unreachable, in which case verification succeeds (lines 3 and 7).
If the $\pi$ state remains reachable by the end of the algorithm, however, then we report a failure (line 8).
The algorithm performs intermediate minimization on lines 1 and 5.
In general, there are many options for which components--or composition of components--to minimize \cite{Sabnani:1989}.
We choose to only minimize components because we observed that minimizing the composition of components was generally slow.
In essence, this algorithm is an abstraction-refinement loop where each new component lowers the abstraction by introducing more state variables.

\begin{example}
    \label{ex:tp-comp-verif}
    Consider $TP$ with ten resource managers and the optimal mapping $f$ from Ex.~\ref{ex:tp}, where $D_P = RM$, $D_1 = Env \parsym TM_1$, and $D_2 = TM_2$.
    Line 1 of Alg.~\ref{alg:comp-verify} will generate an LTS for $D_P$ with 477,454 states, including a $\pi$ state.
    Minimization reduces $D_P$ to 13,291 states.
    Due to a reachable $\pi$ state, the algorithm proceeds into the loop on line 4.
    Next, on line 5, the algorithm generates an LTS for $D_1$ with 3,072 states, which reduces to 1,026 states after minimization.
    Composing this LTS with $\mathcal{D}$ (line 5) retains the $\pi$ state (line 6) so we loop again.
    %The algorithm generates an LTS for $D_2$ with 1,024 states, and minimization provides no reduction.
    %However, after composing this LTS with $\mathcal{D}$ (line 5), the $\pi$ state is no longer reachable and hence we terminate with a positive answer (line 6 and 7).
    The algorithm continues in this fashion until a $\pi$ state is no longer reachable, and we return a positive answer (line 6 and 7).
    A maximum of 481,550 states are needed in memory at once.

    \begin{comment}
    Consider $TP$ with ten resource managers and the optimal mapping $f$ from Ex.~\ref{ex:tp}. %where $f(Env) = f(TM_1)$. %(shown in Fig.~\ref{fig:optimal-recomp-map}).
    %Consider $TP$ with ten resource managers and the optimal mapping $f$ shown in Fig.~\ref{fig:optimal-recomp-map}.
    Recall that, in this case, $D_P = RM$, $D_1 = Env \parsym TM_1$, and $D_2 = TM_2$.
    Line 1 of Alg.~\ref{alg:comp-verify} will generate an LTS for $D_P$ that includes 477,454 states and a $\pi$ state.
    Minimization reduces $D_P$ to 13,291 states.
    The LTS has a reachable $\pi$ state so the algorithm proceeds into the loop on line 4.
    Next, the algorithm will generate an LTS for $D_1$ with 3,072 states, which reduces to 1,026 states after minimization.
    Composing this LTS with $\mathcal{D}$ (line 5) retains the $\pi$ state (line 6) so we loop again.
    The algorithm generates an LTS for $D_2$ with 1,024 states, and minimization provides no reduction.
    However, after composing this LTS with $\mathcal{D}$ (line 5), the $\pi$ state is no longer reachable and hence we terminate with a positive answer (line 6 and 7).
    The maximum number of states needed in memory at once for this algorithm is 481,550.
    \end{comment}
\end{example}
%\vspace{-9pt}
\begin{algorithm}
\centering
\caption{\textsc{Comp-Verify}}
\label{alg:comp-verify}
\renewcommand{\algorithmicrequire}{\textbf{Input:}}
\renewcommand{\algorithmicensure}{\textbf{Output:}}
\begin{algorithmic}[1]
\Require{$D_1,\dots,D_m,D_P,P$}
\Ensure{If $\tolts(D_1) \parexp \dots \parexp \tolts(D_m) \models \prop$}
\State $\mathcal{D} \leftarrow Min(\errlts(D_P,P))$ \Comment{$\prop = \errlts(D_P,P)$}
\If{$\pi \notin Reach(\mathcal{D})$}
    \State\Return $true$
\EndIf
\For{$j \in \{1,\dots,m\}$}
    \State $\mathcal{D} \leftarrow \mathcal{D} \parexp Min(\tolts(D_j))$
    \If{$\pi \notin Reach(\mathcal{D})$}
        \State\Return $true$
    \EndIf
\EndFor
\State\Return $false$
\end{algorithmic}
%\vspace{-3pt}
\end{algorithm}

%% file: src/Tex/45-correctness.tex
\subsection{Correctness Analysis}
\label{sec:correctness-analysis}

In this section, we show that Alg.~\ref{alg:recomp-verify} is sound but not complete.
To establish this result, we first provide lemmas for the correctness of decomposition (Lem.~\ref{lem:decompose-correct}), static specification reduction (Lem.~\ref{lem:static-reduction}), and compositional verification (Lem.~\ref{lem:comp-verif}).
Next, we show that reducing the monolithic model checking problem to compositional verification is correct (Lem.~\ref{lem:mono-to-comp-verif}).
Finally, we present Thm.~\ref{thm:recomp-verif-soundness} that shows that Alg.~\ref{alg:recomp-verify} is sound.
%Finally, soundness (Thm.~\ref{thm:recomp-verif-soundness}) follows as a direct consequence of these three lemmas.
%Therefore, we begin by presenting two lemmas: Lem.~\ref{lem:decompose-correct} for decomposition correctness and Lem.~\ref{lem:comp-verif} for compositional verification correctness.
%Therefore, we first present Lem.~\ref{lem:decompose-correct} and Lem.~\ref{lem:comp-verif} to establish correctness for decomposition and compositional verification respectively.
%Theorems \ref{thm:decompose-correct} and \ref{thm:comp-verif} establish correctness for decomposition and compositional verification respectively; we provide proofs in Appendix~\ref{apx:comp-verif-correctness}.
%The following theorem shows that the decomposition algorithm adheres to properties (P1) and (P2).
%We provide a proof in Appendix~\ref{apx:slicing-procedure}.
\begin{lemma}
    \label{lem:decompose-correct}
    Algorithm~\ref{alg:decompose} ensures
    (P1) $S = C_1 \parsym \dots \parsym C_n$ and
    (P2) $\vbeta C_1 \subseteq \vbeta P$
    upon termination.
\end{lemma}
\begin{proof}
    Sketch.
    We prove property (P1) by establishing the following loop invariant in Alg.~\ref{alg:decompose} on line 5:
    $S = C_1 \parsym \dots \parsym C_i \parsym T_i$.
    The proof for property (P2) follows in two steps.
    First, $\vbeta P \subseteq V_C$ because $V_C$ (in the first partition) is defined by a monotonically increasing operation (\fixpoint) on $\vbeta P$.
    Second, $\vbeta P \subseteq \vbeta C_1$ because $C_1 = \textsc{Slice}(S,V_C)$ will contain exactly the state variables in $V_C$.
\end{proof}

\begin{lemma}
    \label{lem:static-reduction}
    $S \models P$ if and only if $\parsym \left( \bigcup_i X_i \right) \models P$.
\end{lemma}
\begin{proof}
    We prove this theorem in Appendix~\ref{apx:static-reduction}.
\end{proof}

\begin{lemma}
    \label{lem:comp-verif}
    $\tolts(D_1) \parexp \dots \parexp \tolts(D_m) \models \prop$ if and only if there exists a $k \in \{0 \dots m\}$ such that $\pi \notin Reach(\errlts(D_P,P) \parexp \tolts(D_1) \parexp \dots \parexp \tolts(D_{k}))$.
\end{lemma}
\begin{proof}
    Sketch.
    The forwards case ($\Rightarrow$) follows by choosing $k=m$.
    For the backwards case ($\Leftarrow$), we assume that $k$ is given such that $0 \leq k \leq m$ and $\pi \notin Reach(\errlts(D_P,P) \parexp \tolts(D_1) \parexp \dots \parexp \tolts(D_{k}))$.
    Notice that, by construction, none of $\tolts(D_{k+1}), \dots, \tolts(D_m)$ contains a $\pi$ state, and therefore neither will $Reach(\errlts(D_P,P) \parexp \tolts(D_1) \parexp \dots \parexp \tolts(D_{m}))$.
\end{proof}

%Finally, Thm.~\ref{thm:mono-to-comp-verif}--coupled with Lem.~\ref{lem:comp-verif}--shows that Alg.~\ref{alg:recomp-verify} is sound.
\begin{lemma}
    \label{lem:mono-to-comp-verif}
    $S \models P$ $\Longleftrightarrow$ $\tolts(D_1) \parexp \dots \parexp \tolts(D_m) \models \mathcal{P}$.
\end{lemma}
\begin{proof}
    \begin{align}
        &S \models P \\
        \Longleftrightarrow \ &\pi \in Reach(\errlts(S,P)) \\
        \Longleftrightarrow \ &\pi \in Reach(\errlts(C_1 \parsym \dots \parsym C_n,P)) \\
        \Longleftrightarrow \ &\pi \in Reach(\errlts(D_P \parsym D_1 \parsym \dots \parsym D_m,P)) \\
        \Longleftrightarrow \ &\pi \in Reach(\errlts(D_P,P) \parexp \dots \parexp \errlts(D_m,P)) \\
        \Longleftrightarrow \ &\pi \in Reach(\errlts(D_P,P) \parexp \dots \parexp \tolts(D_m)) \\
        \Longleftrightarrow \ &\tolts(D_1) \parexp \dots \parexp \tolts(D_m) \models \proplts(D_P,P)
    \end{align}
    Biconditional (2) holds because $P$ is a safety property,
    (3) by Lem.~\ref{lem:decompose-correct} property (P1),
    (4) by the definition for each $D_j$ (and Lem.~\ref{lem:static-reduction} in the case that $f$ is partial),
    (5) by Thm.~\ref{thm:composition-sem},
    (6) by Lem.~\ref{lem:decompose-correct} property (P2) and Def.~\ref{def:recomp-map} ($f(C_1)=d_P$), and
    (7) because $P$ is a safety property.
    Finally, the theorem follows because $\mathcal{P} = \proplts(D_P,P)$.
\end{proof}

\begin{theorem}
    \label{thm:recomp-verif-soundness}
    If Alg~\ref{alg:recomp-verify} terminates, then it returns true (model checking succeeds) if and only if $S \models P$.
\end{theorem}
\begin{proof}
    The result follows from Lem.~\ref{lem:comp-verif} and Lem.~\ref{lem:mono-to-comp-verif}.
    %The result follows from Lemmas \ref{lem:decompose-correct}, \ref{lem:static-reduction}, \ref{lem:comp-verif}, and \ref{lem:mono-to-comp-verif}.
\end{proof}

%\ek{Are you missing a theorem that says Alg. 1 is sound? It might not be obvious to the reader why Thm 3 and Lem 2 together imply the soundness}

While Thm.~\ref{thm:recomp-verif-soundness} shows that Alg.~\ref{alg:recomp-verify} is sound, the algorithm is not complete, even if we limit $S$ to be a finite-state specification.
This is because a given component $D_j$ may be infinite-state, in which case LTS construction will fail to terminate on line 1 or 5 in Alg.~\ref{alg:comp-verify}.
We address this limitation in Sec.~\ref{sec:portfolio-of-strategies} by using a portfolio of strategies that includes the \textit{monolithic strategy}.

%% file: src/Tex/47-parallel-alg.tex
\section{Choosing Efficient Recomposition Maps}
\label{sec:parallel-alg}

In this section, we address the problem of designing recomposition strategies, i.e. choosing efficient recomposition maps.
Rather than finding a single recomposition strategy, we propose running Alg.~\ref{alg:recomp-verify} with a portfolio of strategies in parallel.
The primary challenge is determining which strategies to use, since the number of possible recomposition maps grows  large as the number of components increases.
We therefore propose a heuristic for pruning the search space of recomposition maps in Sec.~\ref{sec:heuristic}.
We then choose a small portfolio of strategies based on this heuristic in Sec.~\ref{sec:portfolio-of-strategies}.

\subsection{Recomposition Map Reduction Heuristic}
\label{sec:heuristic}
Any heuristic for pruning the search space of recomposition maps should be tailored to the compositional verification algorithm being used.
Since we use a CRA-style verification algorithm, we design our heuristic to find component orderings that can take advantage of short-circuiting.
In particular, the heuristic identifies recomposition maps that order $D_j$ components that are least likely to be necessary for verification \textit{last}.

Our heuristic is to choose recomposition maps that respect the \textit{data flow} partial order $\preccurlyeq$ over the $C_i$ components.
This is a novel partial order that attempts to find dynamic specification reduction--i.e. short-circuiting--by refining our static specification reduction scheme.
Intuitively, the partial order will order the components based on how far removed their state variables are from impacting verification.

More formally, we compute the data flow partial order based on the indexed sets $X_i$ introduced in Sec.~\ref{sec:static-spec-reduction}.
These sets cumulatively capture the components that may interact--either directly or indirectly--with the actions of $C_1$.
First, we build new indexed sets $E_i$ defined as $E_0 = X_0$ and $E_{i+1} = X_{i+1} \setminus X_i$.
While the $X_i$'s are cumulative, each $E_i$ captures only the additional components in each $X_i$.
Intuitively, the components in $E_i$ are $i$ steps removed from affecting the variables in $C_1$, and hence $i$ steps removed from impacting verification (by property (P2) of decomposition).

Second, we build indexed sets $F_i$ that capture the \textit{data flow} from each component in $E_i$ to the components in $E_{i+1}$.
We define $F_0 = \emptyset$ and:
\begin{align*}
    &F_{i+1} = \left\{(C_j,C_k) \left\vert \genfrac{}{}{0pt}{}{C_j \in E_i \text{ and } C_k \in E_{i+1}}{\text{and } \salpha C_j \cap \salpha C_k \neq \emptyset} \right. \right\}
\end{align*}
Finally, let $F = \bigcup_i F_i$; we define the data flow partial order $\preccurlyeq$ to be the reflexive transitive closure of $F$.
In Appendix~\ref{apx:data-flow-po}, we show formally that the partial order refines the static specification reduction scheme from Sec.~\ref{sec:static-spec-reduction}.

\begin{example}
    \label{ex:data-flow-po}
    %Consider $F$ from the construction of the data flow partial order above.
    %For $TP$ and $Consistent$, $F = \{(RM,Env),(Env,TM_1),(Env,TM_2)\}$ and the data flow partial order for $TP$ is the reflexive transitive closure of this set.
    For $TP$ and $Consistent$, $E_0 = \{RM\}$, $E_1 = \{Env\}$, and $E_2 = \{TM_1,TM_2\}$.
    Moreover, $F = \{(RM,Env),(Env,TM_1),(Env,TM_2)\}$ and the data flow partial order is the reflexive transitive closure of this set.
    Intuitively, the partial order shows that $TM_1$ and $TM_2$ can only affect the variables in $RM$--i.e. the variables $\vbeta Consistent$ needed for verification--\textit{indirectly} by interacting with the $Env$ component.
    %Intuitively, the partial order shows that $TM_1$ and $TM_2$ are ``a step removed'' from affecting the variables in $RM$ (i.e. the variables $\vbeta Consistent$ needed for verification) because $TM_1$ and $TM_2$ can only influence $RM$'s variables \textit{indirectly} by interacting with the $Env$ component.
\end{example}

To further reduce the search space of maps, we extend the data flow partial order to a total order $\totorder{}$, i.e. $\preccurlyeq \subseteq \totorder{}$.
We build the total order by breaking ties between incomparable components $C_i$ and $C_j$ by requiring $C_i \totorder{} C_j$ if and only if $C_i$'s state variables have fewer syntactic appearances than $C_j$'s in the original specification $S$.
In the case that the variables of $C_i$ and $C_j$ have the same number of appearances in $S$, we break the tie arbitrarily.

\subsection{Choosing a Portfolio of Strategies}
\label{sec:portfolio-of-strategies}
In this section, we describe four recomposition strategies that comprise our portfolio.
%In the following, we describe the four strategies assuming that the components $C_1, \dots, C_n$ have been reordered according to the total order $\totorder{}$ described in Sec.~\ref{sec:heuristic}.
For simplicity, we describe the strategies assuming that the components $C_1, \dots, C_n$ have been reordered according to the total order $\totorder{}$ described in Sec.~\ref{sec:heuristic}.
The four strategies are (S1) the identity strategy, in which $m=n-1$ and $f(C_i) = d_i$ for all $i$,
(S2) a ``bottom heavy'' strategy in which we choose $m=1$ and $f$ such that $f(C_1) = d_P$ and $f(C_i) = d_1$ for all $i > 1$,
(S3) a ``top heavy'' strategy in which we choose $m=1$ and $f$ such that $f(C_n) = d_1$ and $f(C_i) = d_P$ for all $i < n$, and
(S4) the \textit{monolithic} strategy, where $m=0$ and $f(C_i) = d_P$ for all $i$.
%\ek{The bottom and top heavy strategies seem ``extreme'' in that it groups all the components except one into a single D; is there a reason why these strategies were designed this way? }

\begin{example}
    \label{ex:strategy-portfolio}
    In $TP$, the state variables of $TM_2$ occur fewer times than the variables of $TM_1$.
    Therefore, the total ordering from Sec.~\ref{sec:heuristic} is: $RM,Env,TM_2,TM_1$.
    %Then,
    %(S1) in the identity strategy: $m=3$, $f(RM) = d_P$, $f(Env) = d_1$, $f(TM_2) = d_2$, and $f(TM_1) = d_3$,
    %(S2) the bottom heavy strategy: $m=1$, $f(RM) = d_P$, and $f(Env) = f(TM_2) = f(TM_1) = d_1$,
    %(S3) the top heavy strategy: $m=1$, $f(RM) = f(Env) = f(TM_2) = d_P$, and $f(TM_1) = d_1$, and
    %(S4) the monolithic strategy: $m=0$ and $f(RM) = f(Env) = f(TM_2) = f(TM_1) = d_P$.
    Then:
    \begin{enumerate}
        \item In the identity strategy, $m=3$, $f(RM) = d_P$, $f(Env) = d_1$, $f(TM_2) = d_2$, and $f(TM_1) = d_3$.
        \item In the bottom heavy strategy, $m=1$, $f(RM) = d_P$, and $f(Env) = f(TM_2) = f(TM_1) = d_1$.
        \item In the top heavy strategy, $m=1$, $f(RM) = f(Env) = f(TM_2) = d_P$, and $f(TM_1) = d_1$.
        \item In the monolithic strategy, $m=0$ and $f(RM) = f(Env) = f(TM_2) = f(TM_1) = d_P$.
    \end{enumerate}
\end{example}

As a note regarding the correctness analysis from Sec.~\ref{sec:correctness-analysis}, including the monolithic strategy in the portfolio ensures termination.
Therefore, our parallel approach is complete for finite state specifications $S$.

%% file: src/Tex/50-experiments.tex
\section{Experimental Results}
\label{sec:results}

\subsection{Implementation}
We have created a model checker called Recomp-Verify that can verify safety for \TLA{} specifications.
The model checker is a prototype research tool that implements Alg.~\ref{alg:recomp-verify} in the Python, Java, and Kotlin programming languages.
The model checker also supports running multiple instances of Alg.~\ref{alg:recomp-verify} in parallel, and returns the first result to finish.
Our tool is available in a public repository \cite{Dardik:2024rv}.

\begin{figure*}
    \centering
    \footnotesize
    \input{src/Tables/results}
    \caption{Run time comparison between Recomp-Verify and TLC. %for verifying safety properties in distributed protocols.
        The superscripts for each tool indicates how many threads are allocated to a trial.
        The fastest times for each experiment are bolded.
        The ``Strat.'' column denotes the fastest strategy.}
    \label{fig:verif-results}
\end{figure*}

\subsection{Experiments}
\label{sec:results-experiments}
We evaluate Recomp-Verify against TLC on a benchmark of distributed protocols \cite{schultz:2022}, plus the tla-twophase-counter protocol that we introduce in Ex.~\ref{ex:tp-counter}.
Our evaluation is driven by two research questions.
First, (\textbf{RQ1}) can hand-written recomposition maps provide more efficient verification than TLC?
If this is the case, we then ask whether our technique is still performant when automating the search for recomposition maps.
More precisely, (\textbf{RQ2}) is the performance of Recomp-Verify (using a parallel, portfolio strategy) competitive with TLC when each tool is allotted four threads?

In our experiments, we use TLC$^1$ and TLC$^4$ to respectively denote TLC run with one and four parallel threads; this is a built-in option for the tool.
Recomp-Verify$^1$ is the version of our tool with one thread and hand-crafted maps, while Recomp-Verify$^4$ denotes the version that uses four threads to run the portfolio of recomposition strategies (S1-4) from Sec.~\ref{sec:parallel-alg} in parallel.
We report the fastest strategy for Recomp-Verify$^4$ in the ``Strat.'' column in Fig.~\ref{fig:verif-results}.
Additionally, in the implementation of Recomp-Verify$^4$, we use TLC$^1$ for running the monolithic strategy (S4), since TLC is far more efficient than our research prototype for monolithic model checking.
For example, in ex-quorum-leader-election-6 in Fig.~\ref{fig:verif-results}, Recomp-Verify$^1$ uses an optimal single-threaded strategy, yet is slower than TLC$^1$--and therefore Recomp-Verify$^4$ too.
%For example, in ex-quorum-leader-election-6 in Fig.~\ref{fig:verif-results}, Recomp-Verify$^4$ completes faster than Recomp-Verify$^1$ because the monolithic strategy (S4) 

Every experiment in this paper was run on an Apple MacBook Pro with 32GB of memory and an M1 processor.
For each benchmark, we report the total run time using the Unix \textit{time} utility as well as the maximum number of states checked.
We use TO to indicate a timeout after ten minutes and OM to indicate a program crash due to reaching the memory limit given a 25GB allotment.
For TLC's maximum state count, we use the number of unique states that the tool reports.
For Recomp-Verify, we use the maximum between
(1) the number of unique states generated for each component and,
for each iteration, (2) the number of states that results from composition in Alg.~\ref{alg:comp-verify} on line 5.
For Recomp-Verify$^1$, we also report the number of components that result from decomposition ($n$), the number of recomposed components ($m$), and the number of recomposed components that were checked ($k$).

\subsection{Results and Discussion}
\label{sec:results-discussion}
\subsubsection{RQ1}
We show our results in Fig.~\ref{fig:verif-results}.
In terms of state space, TLC$^1$ enumerates at least as many states as Recomp-Verify$^1$ in every case.
For six of the benchmarks that both tools verified, recomposition reduced the state size by \textit{millions} of states.
Moreover, Recomp-Verify$^1$ short-circuits ($k<m$) for eight benchmarks, each of which has a significantly smaller state space than TLC$^1$.
Finally, Recomp-Verify$^1$--but not TLC$^1$--was able to verify the one infinite state benchmark (tla-twophase-counter) via static specification reduction.

In terms of verification speed, Recomp-Verify$^1$ and TLC$^1$ both outperform each other in fourteen benchmarks, and tie in five cases.
However, Recomp-Verify$^1$ completes more benchmarks, verifying twenty-nine benchmarks while TLC$^1$ verifies twenty-four.
Generally speaking, Recomp-Verify$^1$ is more performant on larger benchmarks; on benchmarks with over a million states, TLC$^1$ is faster in two cases while Recomp-Verify$^1$ is faster in at least six cases.
We therefore answer RQ1 by concluding that hand-crafted maps \textit{can} provide more efficiency than TLC.

\subsubsection{RQ2}
The results for TLC$^4$ and Recomp-Verify$^4$ are similar to the single threaded versions.
In terms of state space, Recomp-Verify$^4$ always enumerates the same number of (or fewer) states than TLC$^4$, and also exhibits large savings in the millions for six benchmarks.
We note that Recomp-Verify$^4$ short-circuited every time that Recomp-Verify$^1$ short-circuited, which showcases the effectiveness of the data flow heuristic from Sec.~\ref{sec:heuristic}.

In terms of verification speed, Recomp-Verify$^4$ is faster for eleven benchmarks, TLC$^4$ is faster for fourteen, and the tools tie for eight benchmarks.
However, Recomp-Verify$^4$ verifies more benchmarks, completing thirty-two, while TLC$^4$ completes twenty-four.
While threading made TLC faster for the smaller benchmarks, it did not help the tool verify more benchmarks.
On the other hand, threading helped Recomp-Verify to verify three more benchmarks. %, two of which have a state space over a million.
Generally speaking, Recomp-Verify$^4$ outperforms TLC$^4$ for large benchmarks and is competitive with TLC$^4$ for smaller ones, and therefore we answer RQ2 in the affirmative.

\subsubsection{Discussion}
In Fig.~\ref{fig:verif-results}, the Strat. column shows that the winning strategy for Recomp-Verify$^4$ varies depending on the given benchmark.
This observation suggests that using a portfolio of strategies may be necessary for efficient verification.
%Furthermore, we observe that the majority of benchmarks (6 out of 8) that short-circuited used the top heavy strategy S3.
%For these cases, program slicing alone would have been sufficient for efficient verification; however, recopmosition is necessary for the two corner cases.
Notably, among the benchmark problems, we found that the bottom heavy strategy (S2) did not perform well.
Most likely, this is because the second component in this strategy is too large, and therefore misses opportunities for short-circuiting.
%Furthermore, we observe that the bottom heavy strategy (S2) does not win once, and therefore could potentially be replaced with a better strategy.

Ultimately, Recomp-Verify tends to be faster than TLC on benchmarks that have more opportunity for recomposition.
We point out that Recomp-Verify$^1$ is faster than TLC$^1$ in \textit{every} case where decomposition produced at least four components ($n \geq 4$).
The same is true for Recomp-Verify$^4$ and TLC$^4$ in all but one case.
This observation suggests that the potential benefits of recomposition increase with the number of available components ($n$).

%% file: src/Tables/results.tex
\rowcolors{2}{gray!10}{white}
\begin{tabular}{l|rrrrr|rr|rrr|rr}
    \multicolumn{0}{c}{} & \multicolumn{5}{c}{Recomp-Verify$^1$} & \multicolumn{2}{c}{TLC$^1$} & \multicolumn{3}{c}{Recomp-Verify$^4$} & \multicolumn{2}{c}{TLC$^4$} \\ \hline
    Name & n & m & k & States & Time & States & Time & States & Time & Strat. & States & Time \\ \hline
    tla-consensus-3
        & 1 & 0 & 0 & 4 & 1s
        & 4 & 1s
        & 4 & 1s & S4
        & 4 & 1s \\
    tla-tcommit-3
        & 1 & 0 & 0 & 34 & 1s
        & 34 & 1s
        & 34 & 1s & S4
        & 34 & 1s \\
    i4-lock-server-2-2
        & 1 & 0 & 0 & 9 & 1s
        & 9 & 1s
        & 9 & 1s & S4
        & 9 & 1s \\
    ex-quorum-leader-election-6
        & 2 & 1 & 1 & 117,671 & 39s
        & 121,111 & \textbf{4s}
        & 121,111 & 5s & S4
        & 121,111 & \textbf{2s} \\
    pyv-toy-consensus-forall-6-6
        & 3 & 1 & 1 & 117,671 & 33s
        & 121,111 & \textbf{4s}
        & 121,111 & 5s & S4
        & 121,111 & \textbf{2s} \\
    tla-simple-5
        & 1 & 0 & 0 & 723 & 1s
        & 723 & 1s
        & 723 & 1s & S4
        & 723 & 1s \\
    ex-lockserv-automaton-20
        & 5 & 1 & 0 & 61 & \textbf{2s}
        & - & TO
        & 61 & \textbf{3s} & S3
        & - & TO \\
    tla-simpleregular-5
        & 1 & 0 & 0 & 2,524 & 2s
        & 2,524 & \textbf{1s}
        & 2,524 & 1s & S4
        & 2,524 & 1s \\
   pyv-sharded-kv-3-3-3
        & 3 & 0 & 0 & 10,648 & 5s
        & 10,648 & \textbf{2s}
        & 10,648 & 2s & S4
        & 10,648 & \textbf{1s} \\
    pyv-lockserv-20
        & 5 & 1 & 0 & 61 & \textbf{1s}
        & - & TO
        & 61 & \textbf{2s} & S3
        & - & TO \\
    tla-twophase-9
        & 4 & 2 & 2 & 145,176 & \textbf{19s}
        & 10,340,352 & 9m41s
        & 145,691 & \textbf{31s} & S1
        & 10,340,352 & 2m36s \\
    tla-twophase-10
        & 4 & 2 & 2 & 481,550 & \textbf{1m8s}
        & - & TO 
        & 482,577 & \textbf{1m36s} & S1
        & - & TO \\
    tla-twophase-counter-9
        & 5 & 2 & 2 & 145,176 & \textbf{19s}
        & - & TO
        & 145,691 & \textbf{31s} & S1
        & - & TO \\
    i4-learning-switch-4-3
        & 1 & 0 & - & - & TO
        & 1,344,192 & \textbf{5m55s}
        & 1,344,192 & 5m55s & S4
        & 1,344,192 & \textbf{1m37s} \\
    ex-simple-decentralized-lock-4
        & 2 & 0 & 0 & 20 & 2s
        & 20 & \textbf{1s}
        & 20 & 1s & S4
        & 20 & 1s \\
    i4-two-phase-commit-7
        & 4 & 2 & 2 & 151,348 & \textbf{26s}
        & 10,016,384 & 3m38s
        & 184,112 & \textbf{27s} & S3
        & 10,016,384 & 53s \\
    pyv-consensus-wo-decide-4
        & 5 & 2 & 1 & 32,816 & \textbf{9s}
        & - & TO
        & 32,953 & \textbf{10s} & S3
        & - & TO \\
    pyv-consensus-forall-4-4
        & 6 & 1 & 0 & 33,545 & \textbf{8s}
        & - & TO
        & 33,545 & \textbf{9s} & S3
        & - & TO \\
    pyv-learning-switch-trans-3
        & 2 & 1 & 0 & 729 & \textbf{5s}
        & - & TO
        & 729 & \textbf{6s} & S1
        & - & TO \\
    pyv-learning-switch-sym-2
        & 2 & 1 & 0 & 4 & 2s
        & 1,344 & \textbf{1s}
        & 1,344 & 1s & S4
        & 1,344 & 1s \\
    pyv-sharded-kv-no-lost-keys-3-3-3
        & 2 & 0 & 0 & 9,261 & 4s
        & 27 & \textbf{1s}
        & 9,261 & 2s & S4
        & 9,261 & \textbf{1s} \\
    ex-naive-consensus-4-4
        & 3 & 1 & 1 & 824 & 2s
        & 1,001 & \textbf{1s}
        & 1,001 & 2s & S4
        & 1,001 & \textbf{1s} \\
    pyv-client-server-ae-4-2-2
        & 2 & 1 & 1 & 352,145 & \textbf{42s}
        & 2,039,392 & 1m36s
        & 352,145 & 49s & S1
        & 2,039,392 & \textbf{28s} \\
    pyv-client-server-ae-2-4-2
        & 2 & 1 & 1 & 894,437 & 2m18s
        & 2,387,032 & \textbf{1m16s}
        & 2,387,032 & 1m26s & S4
        & 2,387,032 & \textbf{22s} \\
    ex-simple-election-6-7
        & 3 & 1 & 0 & 267,590 & \textbf{1m20s}
        & 2,900,256 & 3m7s
        & 267,590 & 1m22s & S3
        & 2,900,256 & \textbf{54s} \\
    pyv-toy-consensus-epr-8-3
        & 3 & 1 & 1 & 65,543 & 1m1s
        & 70,903 & \textbf{6s}
        & 70,903 & 7s & S4
        & 70,903 & \textbf{2s} \\
    ex-toy-consensus-8-3
        & 2 & 1 & 1 & 65,543 & 57s
        & 70,903 & \textbf{5s}
        & 70,903 & 6s & S4
        & 70,903 & \textbf{2s} \\
    pyv-client-server-db-ae-2-3-2
        & 5 & 4 & 4 & 188,158 & \textbf{12s}
        & 1,394,368 & 1m1s
        & 188,799 & \textbf{15s} & S1
        & 1,394,368 & 18s \\
    pyv-client-server-db-ae-4-2-2
        & 5 & 1 & 1 & 356,706 & \textbf{1m23s}
        & 3,624,960 & 2m48s
        & 356,706 & 1m40s & S1
        & 3,624,960 & \textbf{44s} \\
    pyv-firewall-5
        & 2 & 0 & 0 & 56,072 & 9s
        & 56,072 & \textbf{2s}
        & 56,072 & 3s & S4
        & 56,072 & \textbf{1s} \\
    ex-majorityset-leader-election-5
        & 3 & 1 & - & - & TO
        & 166,306 & \textbf{15s}
        & 166,306 & 17s & S4
        & 166,306 & \textbf{5s} \\
    pyv-consensus-epr-4-4
        & 6 & 2 & 1 & 7,018 & \textbf{3s}
        & - & TO
        & 7,221 & \textbf{5s} & S3
        & - & TO \\
    mldr-2
        & 1 & 0 & - & - & TO
        & - & TO
        & - & TO & -
        & - & TO
\end{tabular}

%% file: src/Tex/60-related-work.tex
\section{Related Work}
Compositional verification is a well studied research area.
Two widely studied styles of compositional verification are CRA \cite{Cheung:1995,Cheung:1996,Cheung:1999,Graf:1990,Malhotra:1990,Sabnani:1989,Tai:1993,Zheng:2010,Zheng:2014,Zheng:2015} and assume-guarantee reasoning \cite{alur:2005,Chen:2009,Cobleigh:2003,Cobleigh:2006,Gupta:2007,Jones:1983,Metzler:2008,Nam:2008,Pasareanu:2008,Pasareanu:2006,Pnueli:1985}, although other styles exist as well \cite{Barringer:2003,Bensalem:2008}.
Of these works, the ones most closely related to this paper automate decomposition for verification.
Metzler et a. \cite{Metzler:2008} and Cobleigh et al. \cite{Cobleigh:2006} decompose systems into two components, after which they apply $\text{L}^*$ style learning \cite{Angluin:1987} to infer assumptions for assume-guarantee style compositional verification.
Nam et al. \cite{Nam:2008} use a similar strategy, but consider multi-way decomposition and verification.
While these works report limited success, we are able to find efficient verification problems via recomposition.

Our work also relates to program slicing \cite{Bruckner:2005,Weiser:1982} and cone of influence reduction \cite{Clarke:2018model}, both of which are techniques for static specification reduction.
These two techniques soundly reduce the state variables needed for model checking by analyzing a variable dependency graph.
Our work includes static specification reduction by allowing partial recomposition maps, as described in Sec.~\ref{sec:algorithm}.

In the \TLA{} ecosystem, TLC \cite{Yu:1999} is the most well-known model checker.
Apalache \cite{Konnov:2019,Otoni:2023} is an alternate model checker that internally relies on SMT solvers.
Apalache supports bounded model checking and verification with inductive invariants--two techniques that are outside the scope of comparison for our tool.
The \TLA{} Proof System (TLAPS) \cite{cousineau:2012tla} provides an alternative to model checking \TLA{}.
TLAPS proofs are manually constructed, but automatically verified by dispatching proof obligations to SMT solvers.
Endive \cite{schultz:2022} is a tool that automatically infers inductive invariants for \TLA{} specifications, which then may be checked using a TLAPS proof.

%% file: src/Tex/70-conclusion.tex
\section{Limitations and Future Work}
In Sec.~\ref{sec:results-discussion}, we show that the effectiveness of our approach is tied to the number of $C_i$ components.
In future work, we plan to investigate methods to make decomposition more granular, as well as decomposing \textit{properties}.
As the number of components increase, we also plan to improve our methods for finding efficient recomposition maps.
For example, we plan to improve our parallel technique so that different threads can share intermediate work to save time and memory. %; this may allow us to explore even more recomposition maps at once.

In this paper, we focus on using recomposition for explicit-state model checking. However, recomposition may also be effective for other compositional verification tasks.
For example, we plan to investigate whether recomposition can be used in combination with non-explicit verification techniques, e.g. using SMT solvers or inductive invariants.
We also plan to investigate whether recomposition can be used for efficient counter-example detection.

%% file: src/Tex/90-appendix.tex
\section{Proof of Thm.~\ref{thm:composition-sem}}
\label{apx:composition-sem-proof}

Before proving Thm~\ref{thm:composition-sem}, we first introduce more general semantics for \TLA{} specifications, namely action-state-based semantics.
We define an action-state-based behavior $\sigma$ as an infinite sequence of action-state pairs: $(\epsilon,b_0)(a_1,b_1)(a_2,b_2)\dots$, where $\epsilon$ is the empty action.
We use the notation $\sigma_i$ to access the $i$th pair of $\sigma$.
Notice that we require the initial action to be empty, hence $\sigma_0 = (\epsilon,b_0)$.
Let $\phi$ be an arbitrary \TLA{} formula, let $\psi$ be a \TLA{} formula without temporal operators, and let $Next$ be the transition relation for a specification $S$, i.e. $Next$ is a \TLA{} formula with prime operators.
We now define the action-state-based semantics of a behavior $\sigma$ on a \TLA{} formula in terms of the action-state-based satisfaction operator $\abmodels$.
\begin{align*}
    &\sigma \abmodels \psi &\text{iff}& \quad\quad b_0 \models \psi \\
    &\sigma \abmodels \Box\phi &\text{iff}& \quad\quad \forall i \in \mathbb{N} \st \sigma_i \sigma_{i+1} \dots \abmodels \phi \\
    &\sigma \abmodels \Diamond\phi &\text{iff}& \quad\quad \exists i \in \mathbb{N} \st \sigma_i \sigma_{i+1} \dots \abmodels \phi \\
    &\sigma \abmodels [Next]_v &\text{iff}& \quad\quad \big(a_1 \in \calpha S \land b_1' \models (b_0 \land a_1) \big) \ \lor \\
    &&& \quad\quad \big(a_1 \notin \calpha S \land b_1' \models (b_0 \land v' = v) \big)
\end{align*}
Note that the semantics for $\sigma \abmodels [Next]_v$ implicitly uses $Next$ because $\calpha S$ is defined by $Next$.
We now define the action-state-based semantics of a \TLA{} specification $S$ as the set of all action-state-based behaviors that satisfy $S!Spec$.
More formally, we define the action-state-based semantics of $S$ to be: $\abSemantics{S} = \{ \sigma \mid \sigma \abmodels S!Spec \}$.
We also define $actions(\sigma)$ to denote the action-based behavior $a_1 a_2 \dots$.
We overload this operator on a set of action-state-based behaviors $A$ as follows: $actions(A) = \{ actions(\sigma) \mid \sigma \in A \}$.

\begin{lemma}
    \label{lem:sem-intersection}
    $\abSemantics{S \parsym T} = \abSemantics{S} \cap \abSemantics{T}$
\end{lemma}
\begin{proof}
    We will first show that $\abSemantics{S \parsym T} \subseteq \abSemantics{S} \cap \abSemantics{T}$.
    Suppose that $\sigma \in \abSemantics{S \parsym T}$, then it suffices to show that $\sigma \abmodels S!Spec$ and $\sigma \abmodels T!Spec$.
    %because this implies $\sigma \in \abSemantics{S} \cap \abSemantics{T}$.
    Recall that $Spec$ has the form $Init \land \Box[Next]_{vars}$.
    Hence, we break the proof obligation into two tasks, one for each conjunct.
    The first task is to show that $\sigma \abmodels S!Init$ and $\sigma \abmodels T!Init$; this fact follows because $\sigma \abmodels (S!Init \land T!Init)$ by the definition of $\parsym$.
    The second task is to show that $\sigma \abmodels \Box[S!Next]_{vars}$ and $\sigma \abmodels \Box[T!Next]_{vars}$.
    By the definition of the action-state-based semantics for \TLA{}, this is equivalent to showing (for just $S$):
    \begin{align}
    \label{rule:sem-intersect-1}
    \begin{split}
        \forall i \in \mathbb{N} \st
        &\lor \big( a_{i+1} \in \calpha S \land b_{i+1}' \models (b_i \land S!a_{i+1}) \big) \\
        &\lor \big( a_{i+1} \notin \calpha S \land b_{i+1}' \models (b_i \land S!vars' = S!vars) \big)
        %&\forall i \in \mathbb{N} \st \big(a_{i+1} \in \calpha S \land b_{i+1}' \models (b_i \land a_{i+1}) \big) \lor \big(a_{i+1} \notin \calpha S \land b_{i+1}' \models (b_i \land S!v' = S!v) \big)
    \end{split}
    \end{align}
    We will prove equation (\ref{rule:sem-intersect-1})--for both $S$ and $T$--by considering the following four cases for an arbitrary choice of $i$:
    \begin{enumerate}
        \item Case 1: $a_{i+1} \in \calpha S$ and $a_{i+1} \in \calpha T$.
        
            In this case, the symbolic action for $a_{i+1}$ must be in $\salpha S$ and $\salpha T$, and therefore $(S \parsym T)!a_{i+1}$ is equal to $S!a_{i+1} \land T!a_{i+1}$.
            Thus we see that $b_{i+1}' \models b_i \land S!a_{i+1} \land T!a_{i+1}$ by the action-state-based semantics of \TLA{}.
            This implies equation (\ref{rule:sem-intersect-1}) for the base case for both $S$ and $T$.
            
        \item Case 2: $a_{i+1} \in \calpha S$ and $a_{i+1} \notin \calpha T$.
        
            In this case, the symbolic action for $a_{i+1}$ must be in $\salpha S$ but not in $\salpha T$.
            Therefore, $(S \parsym T)!a_{i+1}$ is equal to $S!a_{i+1} \land (T!vars' = T!vars)$.
            Thus we see that $b_{i+1}' \models b_i \land S!a_{i+1} \land (T!vars' = T!vars)$ by the action-state-based semantics of \TLA{}.
            This implies equation (\ref{rule:sem-intersect-1}) for the base case for both $S$ and $T$.
            
        \item Case 3: $a_{i+1} \notin \calpha S$ and $a_{i+1} \in \calpha T$.

            This case is identical to Case 2, except with $S$ and $T$ flipped.
            
        \item Case 4: $a_{i+1} \notin \calpha S$ and $a_{i+1} \notin \calpha T$.
        
            In this case, the symbolic action for $a_{i+1}$ must not be in $\salpha S$ nor $\salpha T$.
            Therefore, $(S \parsym T)!a_{i+1}$ is equal to $(S!vars' = S!vars) \land (T!vars' = T!vars)$.
            Thus we see that $b_{i+1}' \models b_i \land (S!vars' = S!vars) \land (T!vars' = T!vars)$ by the action-state-based semantics of \TLA{}.
            This implies equation (\ref{rule:sem-intersect-1}) for the base case for both $S$ and $T$.
    \end{enumerate}
    
    Conversely, we will now show that $\abSemantics{S} \cap \abSemantics{T} \subseteq \abSemantics{S \parsym T}$.
    Suppose that $\sigma \in \abSemantics{S} \cap \abSemantics{T}$, then it suffices to show that $\sigma \abmodels (S \parsym T)!Spec$.
    Recall that $Spec$ has the form $Init \land \Box[Next]_{vars}$.
    Hence, we break the proof obligation into two tasks, one for each conjunct.
    The first task is to show $\sigma \abmodels (S \parsym T)!Init$, or equivalently $\sigma \abmodels S!Init \land T!Init$.
    Because $\sigma \abmodels S!Spec$ and $\sigma \abmodels T!Spec$, we see that $\sigma \abmodels S!Init$ and $\sigma \abmodels T!Init$; thus the first task is proved.
    The second task is to show that $\sigma \abmodels \Box[(S \parsym T)!Next]_{vars}$.
    By the definition of the action-state-based semantics for \TLA{}, this is equivalent to showing:
    \begin{align}
    \label{rule:sem-intersect-2}
    \begin{split}
        &\forall i \in \mathbb{N} \st \\
        &\lor \big( a_{i+1} \in \calpha (S \parsym T) \land b_{i+1}' \models (b_i \land (S \parsym T)!a_{i+1}) \big) \\
        &\lor \big( a_{i+1} \notin \calpha (S \parsym T) \land b_{i+1}' \models \\
        &\quad\quad\quad (b_i \land (S \parsym T)!vars' = (S \parsym T)!vars) \big)
    \end{split}
    \end{align}
    We will prove equation (\ref{rule:sem-intersect-2}) by considering the following four cases for an arbitrary choice of $i$:
    \begin{enumerate}
        \item Case 1: $a_{i+1} \in \calpha S$ and $a_{i+1} \in \calpha T$.
        
            In this case, $a_{i+1} \in \calpha (S \parsym T)$.
            Next, notice that the symbolic action for $a_{i+1}$ must be in $\salpha S$ and $\salpha T$, and therefore $(S \parsym T)!a_{i+1}$ is equal to $S!a_{i+1} \land T!a_{i+1}$.
            Furthermore, $b_{i+1}' \models b_i \land S!a_{i+1}$ and $b_{i+1}' \models b_i \land T!a_{i+1}$.
            Therefore we conclude $b_{i+1}' \models b_i \land S!a_{i+1} \land T!a_{i+1}$ which implies equation (\ref{rule:sem-intersect-2}) for this case.
            
        \item Case 2: $a_{i+1} \in \calpha S$ and $a_{i+1} \notin \calpha T$.
        
            In this case, $a_{i+1} \in \calpha (S \parsym T)$.
            Next, notice that the symbolic action for $a_{i+1}$ must be in $\salpha S$ but not in $\salpha T$.
            Therefore, $(S \parsym T)!a_{i+1}$ is equal to $S!a_{i+1} \land (T!vars' = T!vars)$.
            Furthermore, $b_{i+1}' \models b_i \land S!a_{i+1}$ and $b_{i+1}' \models b_i \land (T!vars' = T!vars)$.
            Therefore we conclude $b_{i+1}' \models b_i \land S!a_{i+1} \land (T!vars' = T!vars)$ which implies equation (\ref{rule:sem-intersect-2}) for this case.
            
        \item Case 3: $a_{i+1} \notin \calpha S$ and $a_{i+1} \in \calpha T$.

            This case is identical to Case 2, except with $S$ and $T$ flipped.
            
        \item Case 4: $a_{i+1} \notin \calpha S$ and $a_{i+1} \notin \calpha T$.
        
            In this case, $a_{i+1} \notin \calpha (S \parsym T)$.
            Next, notice that the symbolic action for $a_{i+1}$ must not be in $\salpha S$ nor $\salpha T$.
            Therefore, $(S \parsym T)!a_{i+1}$ is equal to $(S!vars' = S!vars) \land (T!vars' = T!vars)$.
            Furthermore, $b_{i+1}' \models b_i \land (S!vars' = S!vars)$ and $b_{i+1}' \models b_i \land (T!vars' = T!vars)$.
            Therefore we conclude $b_{i+1}' \models b_i \land ((S!vars \circ T!vars)' = (S!vars \circ T!vars))$.
            Since $(S \parsym T)!vars = (S!vars \circ T!vars)$, this implies equation (\ref{rule:sem-intersect-2}) for this case.
    \end{enumerate}
\end{proof}
Lem.~\ref{lem:sem-intersection} shows that our notion of parallel composition is semantically an intersection in the action-state-based semantics of \TLA{}.
%Many prior works have noted semantic intersection as a desired property of parallel composition \cite{Abadi:1991,Abadi:1995,Abramsky:1993,Zave:1993}.
We now prove one more lemma before proving Thm.~\ref{thm:composition-sem}.

\begin{lemma}
    \label{lem:alpha-intersect}
    Let $\mathcal{A}$ and $\mathcal{B}$ be LTSs, then $\aSemantics{\mathcal{A} \parexp \mathcal{B}} = \aSemantics{\mathcal{A}} \cap \aSemantics{\mathcal{B}}$.
\end{lemma}
\begin{proof}
    We first prove that $\aSemantics{\mathcal{A} \parexp \mathcal{B}} \subseteq \aSemantics{\mathcal{A}} \cap \aSemantics{\mathcal{B}}$.
    Suppose that $\sigma \in \aSemantics{\mathcal{A} \parexp \mathcal{B}}$.
    Then there exists a sequence of state pairs $(p_0,q_0),(p_1,q_1),\dots$ such that $p_0 \in I_{\mathcal{A}}$, $q_0 \in I_{\mathcal{B}}$ and, for each nonnegative index $i$, either
    (1) $\sigma_i \in \alpha \mathcal{A}$, $(p_i,\sigma_i,p_{i+1}) \in \delta_{\mathcal{A}}$, $\sigma_i \in \alpha \mathcal{B}$, and $(q_i,\sigma_i,q_{i+1}) \in \delta_{\mathcal{B}}$,
    (2) $\sigma_i \in \alpha \mathcal{A}$, $(p_i,\sigma_i,p_{i+1}) \in \delta_{\mathcal{A}}$, $\sigma_i \notin \alpha \mathcal{B}$, and $q_i = q_{i+1}$,
    (3) $\sigma_i \notin \alpha \mathcal{A}$, $p_i = p_{i+1}$, $\sigma_i \in \alpha \mathcal{B}$, and $(q_i,\sigma_i,q_{i+1}) \in \delta_{\mathcal{B}}$, or
    (4) $\sigma_i \notin \alpha \mathcal{A}$, $p_i = p_{i+1}$, $\sigma_i \notin \alpha \mathcal{B}$, and $q_i = q_{i+1}$.
    In each of these four cases, $\sigma_i$ is a valid action from $p_i$ to $p_{i+1}$ in $\mathcal{A}$ and also a valid action from $q_i$ to $q_{i+1}$ in $\mathcal{B}$.
    Given that $p_0 \in I_{\mathcal{A}}$ and $q_0 \in I_{\mathcal{B}}$, we have $\sigma \in \aSemantics{\mathcal{A}}$ and $\sigma \in \aSemantics{\mathcal{B}}$.

    Conversely, we will now show that $\aSemantics{\mathcal{A}} \cap \aSemantics{\mathcal{B}} \subseteq \aSemantics{\mathcal{A} \parexp \mathcal{B}}$.
    Suppose that $\sigma \in \aSemantics{\mathcal{A}} \cap \aSemantics{\mathcal{B}}$.
    Then there exists a sequence of states $p_0,p_1,\dots$ such that $p_0 \in I_{\mathcal{A}}$ and, for each nonnegative index $i$, either
    (1) $\sigma_i \in \alpha \mathcal{A}$ and $(p_i,\sigma_i,p_{i+1}) \in \delta$, or
    (2) $\sigma_i \notin \alpha \mathcal{A}$ and $p_i = p_{i+1}$.
    Likewise, there exists a sequence of state $q_0,q_1,\dots$ with the same conditions but for $\mathcal{B}$.
    Then we can construct a sequence of state pairs $(p_0,q_0),(p_1,q_1),\dots$ such that $p_0 \in I_{\mathcal{A}}$, $q_0 \in I_{\mathcal{B}}$ and, for each nonnegative index $i$, at least one of the following four conditions holds:
    (1) $\sigma_i \in \alpha \mathcal{A}$, $(p_i,\sigma_i,p_{i+1}) \in \delta_{\mathcal{A}}$, $\sigma_i \in \alpha \mathcal{B}$, and $(q_i,\sigma_i,q_{i+1}) \in \delta_{\mathcal{B}}$,
    (2) $\sigma_i \in \alpha \mathcal{A}$, $(p_i,\sigma_i,p_{i+1}) \in \delta_{\mathcal{A}}$, $\sigma_i \notin \alpha \mathcal{B}$, and $q_i = q_{i+1}$,
    (3) $\sigma_i \notin \alpha \mathcal{A}$, $p_i = p_{i+1}$, $\sigma_i \in \alpha \mathcal{B}$, and $(q_i,\sigma_i,q_{i+1}) \in \delta_{\mathcal{B}}$, or
    (4) $\sigma_i \notin \alpha \mathcal{A}$, $p_i = p_{i+1}$, $\sigma_i \notin \alpha \mathcal{B}$, and $q_i = q_{i+1}$.
    Thus, by definition, $\sigma \in \aSemantics{\mathcal{A} \parexp \mathcal{B}}$.
\end{proof}
Using the prior two lemmas, we now prove Thm.~\ref{thm:composition-sem}.\\

\begin{thmn}{\ref{thm:composition-sem}}
    $\aSemantics{\tolts(S \parsym T)} = \aSemantics{\tolts(S) \parexp \tolts(T)}$.
\end{thmn}
\begin{proof}
    \begin{align*}
        \aSemantics{\tolts(S \parsym T)} &= actions\left(\abSemantics{S \parsym T}\right)\\
        &= actions\left(\abSemantics{S} \cap \abSemantics{T}\right)\\
        &= actions\left(\abSemantics{S}\right) \cap actions\left(\abSemantics{T}\right)\\
        &= \aSemantics{\tolts(S)} \cap \aSemantics{\tolts(T)}\\
        &= \aSemantics{\tolts(S) \parexp \tolts(T)}
    \end{align*}
    Where the second equality is due to Lem.~\ref{lem:sem-intersection}, the third equality holds because $S$ and $T$ do not share state variables, and the fifth equality is due to Lem.~\ref{lem:alpha-intersect}.
\end{proof}

\section{Slicing Procedure}
\label{apx:slicing-procedure}

We show the procedure for slicing in Fig.~\ref{fig:slicing-procedures}.
Slicing works by creating a new specification $T$ such that $\vbeta T = V$.
Line 2 in Fig.~\ref{fig:slicing-procedures} defines $T$'s state variables to be $S!vars$ restricted to the variables in $V$, while line 3 defines $T$'s initial state predicate to be $S!Init$ restricted to the variables in $V$.
Line 4 defines $acts$ to be the set of the actions in $S$ restricted to the variables in $V$, and line 5 defines $T$'s transition relation to be subset of $acts$ that are well-formed.

\begin{figure}
    %\caption{\textsc{Slice}}
    %\label{alg:slice}
    %\begin{center}
    \begin{algorithmic}[1]
    \Procedure{Slice}{$S,V$}
    \State $T!vars \leftarrow \expr{S!vars \cap V}$
    \State $T!Init \leftarrow \textsc{Slice-Rec}(S!Init,V)$
    \State $acts \leftarrow \{ \textsc{Slice-Rec}(A,V) \mid A \in \salpha S \}$
    \State $T!Next \leftarrow \expr{\bigvee \{ A \in acts \mid A \neq \text{DEL} \}}$
    \State \Return $T$
    \EndProcedure
    %{\tlatex \@pvspace{6.0pt}}
    \Procedure{Slice-Rec}{$e,V$}
    \State \textbf{match} $e$ \textbf{with}
    \State $\expr{op(p) \triangleq e_1} \rightarrow$
        \IndState $E \leftarrow \textsc{Slice-Rec}(e_1,V)$
        \IndState \textbf{if} $E = DEL$ \textbf{then}
            \IndDubState \Return \text{DEL}
        \IndState \textbf{else}
            \IndDubState \Return $\expr{op(p) \triangleq E}$
    \vfill\null
    \columnbreak
    \State $\expr{\exists x \in d \st e_1} \rightarrow$
        \IndState $E \leftarrow \textsc{Slice-Rec}(e_1,V)$
        \IndState \textbf{if} $E = \text{DEL}$ \textbf{then}
            \IndDubState \Return \text{DEL}
        \IndState \textbf{else}
            \IndDubState \Return $\expr{\exists x \in d \st E}$
    \State $\exprb{\bigwedge\limits_{i\in I} c_i} \rightarrow$
        \IndState $C \leftarrow \{ c_j \mid \vbeta c_j \subseteq V \}$
        \IndState \textbf{if} $C = \emptyset$ \textbf{then}
            \IndDubState \Return \text{DEL}
        \IndState \textbf{else}
            %\IndDubState $J \leftarrow \{ j \in I \st C_j \neq \text{DEL} \}$
            %\IndDubState \Return $\exprb{\bigwedge\limits_{j\in J} C_j}$
            \IndDubState \Return $\exprb{\bigwedge C_j}$
    \EndProcedure
    %{\tlatex \@pvspace{6.0pt}}
    \end{algorithmic}
    %\end{center}
    \caption{Slicing procedure definitions.}
    \label{fig:slicing-procedures}
\end{figure}

The \textsc{Slice-Rec} procedure on line 7 recursively descends into a \TLA{} formula to restrict it to a set of variables $V$.
The procedure matches three cases (line 8): definitions (line 9), existential quantifiers (line 15), and conjunctions (line 21).
We use the notation $\epsilon ( . )$ to denote \TLA{} syntax, and $e_1$ (lines 9 and 15) and $c_i$ (line 21) both refer to arbitrary \TLA{} formulas.
In a nutshell, slicing works by removing any conjunct that contains a variable not in $V$ (lines 22 and 26).
In the case that \textit{every} conjunct of a \TLA{} expression must be removed, we consider the entire expression malformed and we mark it for removal.
We mark an expression for removal using the DEL keyword on line 24, which is then bubbled to the top (lines 12 and 18); the actual removal happens in the set comprehension on line 5.

\section{Static Specification Reduction}
\label{apx:static-reduction}
We begin by showing that the sequence $\{X_i\}$ converges by the $n+1^\text{st}$ term, i.e. at $X_n$.
%This is important because it implies that $F = F_0 \cup \dots \cup F_n$, and hence calculating the data flow partial order is decidable.
\begin{theorem}
    \label{thm:xis-bound}
    $X_n = X_{n+1}$
\end{theorem}
\begin{proof}
    First, we note that $n+1 > |X_{n+1}|$ because there are at most $n$ components that can be in each $X_i$.
    Next, a simple inductive argument shows that whenever $X_i \neq X_{i+1}$, $|X_{i+1}| \geq i+1$.
    Finally, for the sake of contradiction, suppose that $X_n \neq X_{n+1}$.
    Then it would be the case that $n+1 > |X_{n+1}| \geq n+1$ which is a contradiction.
\end{proof}

Before presenting Lem.~\ref{lem:static-reduction}, we define $states(\sigma)$ on an action-state-based behavior $\sigma$ to denote the state-based behavior $b_0 b_1 \dots$.
We overload this operator on a set of action-state-based behaviors $B$ as follows: $states(B) = \{ states(\sigma) \mid \sigma \in B \}$.\\

\begin{lemn}{\ref{lem:static-reduction}}
    $S \models P$ if and only if $\parsym \left( \bigcup_i X_i \right) \models P$.
\end{lemn}
\begin{proof}
    Let $N = \parsym (\bigcup_i X_i)$ be the composition of the necessary components and let $U = \parsym (\{C_1,\dots,C_n\} - \bigcup_i X_i)$ be the composition of the unnecessary components.
    We note that $\salpha N \cap \salpha U = \emptyset$, which follows from the definition of $X_{i+1}$.
    Also, because $N$ is composed of $C_1$, $\vbeta P \subseteq \vbeta N$ by property (P2) of Lem.~\ref{lem:decompose-correct}.
    Because $\vbeta P \subseteq \vbeta N$ and components do not share state variables, we can therefore conclude that $\vbeta P \cap \vbeta U = \emptyset$.
    Now we have:
    \begin{align*}
        S \models P &\Longleftrightarrow N \parsym U \models P\\
        &\Longleftrightarrow states(\abSemantics{N \parsym U}) \subseteq \bSemantics{P}\\
        &\Longleftrightarrow states(\abSemantics{N} \cap \abSemantics{U}) \subseteq \bSemantics{P}\\
        &\Longleftrightarrow states(\abSemantics{N}) \subseteq \bSemantics{P}\\
        &\Longleftrightarrow N \models P
    \end{align*}
    In the equation above, the third biconditional is due to Lem.~\ref{lem:sem-intersection}.
    The fourth biconditional follows for two reasons.
    The first reason is that $\vbeta P \cap \vbeta U = \emptyset$, and therefore the states in each behavior of $\abSemantics{U}$ will not restrict the values of the variables in $\vbeta P$.
    The second reason is that $\salpha N \cap \salpha U = \emptyset$; this implies $\calpha N \cap \calpha U = \emptyset$, so the actions in each behavior of $\abSemantics{U}$ will not restrict any behavior in $\abSemantics{N}$.
\end{proof}

\section{The Data Flow Partial Order Refines Static Specification Reduction}
\label{apx:data-flow-po}

The following theorem shows that the data flow partial order refines our static specification reduction scheme by providing \textit{more} information.
In other words, the data flow partial order identifies a class of components that are necessary for verification ($\bigcup_i X_i$) \textit{and} provides an ordering over these components.

\begin{theorem}
    \label{thm:data-flow-po-property}
    Let $G = \{ C_i \mid \exists C_j : C_i \preccurlyeq C_j \text{ or } C_j \preccurlyeq C_i\}$ be the set of all specifications for which the data flow partial order is defined.
    Then $G = \bigcup_i X_i$, i.e. it is exactly the set of necessary components in static specification reduction.
\end{theorem}
\begin{proof}
    Based on the definition for $X_{i+1}$, we provide an alternate formulation for $E_{i+1}$:
    \begin{align*}
        E_{i+1} &= \{ C_j \mid C_j \notin X_i \text{ and } \salpha C_j \cap \salpha E_i \neq \emptyset \}\\
        &= \{ C_j \mid C_j \notin X_i \} \cap \{ C_j \mid \salpha C_j \cap \salpha E_i \neq \emptyset \}
    \end{align*}

    We now provide a formula for the set of all first (lesser) elements in each pair for $F_{i+1}$.
    We will refer to this set as $proj_1(F_{i+1})$ for the projection onto the first element.
    \begin{align*}
        proj_1(F_{i+1}) &= \{C_j \mid C_j \in E_i \text{ and} \\
        &\quad\quad\quad\quad \exists C_k \in E_{i+1} : \salpha C_j \cap \salpha C_k \neq \emptyset \}\\
        &=E_i \cap \{C_j \mid \salpha C_j \cap \salpha E_{i+1} \neq \emptyset \}\\
        &=E_i
    \end{align*}
    Where the final equality is due to the alternate formulation of $E_{i+1}$.
    %Therefore, $\bigcup_{i \geq 0} E_i = \bigcup_{i \geq 1} proj_1(F_i)$.
    %We can further conclude $\bigcup_i E_i = \bigcup_i proj_1(F_i)$ because $F_0 = \emptyset$.

    We now provide a formula for the set of all second (greater) elements in each pair for $F_{i+1}$.
    We will refer to this set as $proj_2(F_{i+1})$ for the projection onto the second element.
    \begin{align*}
        proj_2(F_{i+1}) &= \{C_k \mid C_k \in E_{i+1} \text{ and} \\
        &\quad\quad\quad\quad \exists C_j \in E_i : \salpha C_j \cap \salpha C_k \neq \emptyset \}\\
        &=\{C_k \mid C_k \in E_{i+1} \text{ and } \salpha E_i \cap \salpha C_k \neq \emptyset \}\\
        &=E_{i+1} \cap \{C_k \mid \salpha E_i \cap \salpha C_k \neq \emptyset \}\\
        &=E_{i+1}
    \end{align*}
    Where the final equality is due to the alternate formulation of $E_{i+1}$.
    Putting it all together we see:
    \begin{align*}
        \bigcup_{i \geq 0} X_i &= \bigcup_{i \geq 0} E_i = \left( \bigcup_{i \geq 0} E_i \right) \cup \left( \bigcup_{i \geq 0} E_{i+1} \right)\\
        &= \left( \bigcup_{i \geq 0} proj_1(F_{i+1}) \right) \cup \left( \bigcup_{i \geq 0} proj_2(F_{i+1}) \right)\\
        &= \left( \bigcup_{i \geq 0} proj_1(F_i) \right) \cup \left( \bigcup_{i \geq 0} proj_2(F_i) \right) = G
    \end{align*}
    Where the penultimate equality follows because $F_0 = \emptyset$.
\end{proof}